%% file: main.tex
\documentclass[journal]{IEEEtran}
\ifCLASSINFOpdf
\else
\fi

\input{accessories/packages}
\input{accessories/acro}

\begin{document}
\input{accessories/title}

\input{Sections/abstract}

\input{Sections/notations}

\input{Sections/introduction}

\input{Sections/generic_cov_fitting}

\input{Sections/sequential_cov_fitting}

\input{Sections/experiments}

\input{Sections/real_data}

\input{Sections/acknowledgment}

\input{Sections/conclusion}

\input{biblio}

\clearpage

\input{Sections/appendix/appendix_2cols}
\end{document}

%% file: accessories/packages.tex
\usepackage{amsmath}
\usepackage{hyperref}
\usepackage{tikz}
\usepackage{multirow}
\usepackage{amssymb}
\usepackage{pst-solides3d}
\usepackage{pgfplots}
\pgfplotsset{compat=1.18}
\usepackage{graphicx}
\usepackage{comment}
\usepackage{tabularx}
\usepackage{cite}
\usepackage{float}
\usepackage{algorithm}
\usepackage{algpseudocode}
\usepackage{acro}
\usepackage{amsfonts}
\usepackage{url}
\usepackage{titlesec}
\usepackage{mathrsfs}
\usepackage{pgfplots}
\pgfplotsset{compat=newest}
\usepackage[font=small, labelfont=bf]{caption}
\usepackage[standard]{ntheorem}
\usepackage[numbers, sort&compress]{natbib}
\usepackage{booktabs}
\usepackage{titlesec}
\usepackage{makecell}

%% file: accessories/acro.tex
\DeclareAcronym{SCM}{short=\textsc{SCM}, long = Sample Covariance Matrix}
\DeclareAcronym{InSAR}{short=InSAR, long = Interferometric Synthetic Aperture Radar}
\DeclareAcronym{PO}{short=\textsc{PO}, long = Phase-Only Sample Covariance Matrix}
\DeclareAcronym{KL}{short=\textsc{KL}, long = Kullback-Leibler}
\DeclareAcronym{PL}{short=\textsc{PL}, long = Phase-Linking}
\DeclareAcronym{MLE}{short=\textsc{MLE}, long = Maximum Likelihood Estimator}
\DeclareAcronym{SAR}{short=\textsc{SAR}, long = Synthetic Aperture Radar}
\DeclareAcronym{CCG}{short=\textsc{CCG}, long = Complex Circular Gaussian}
\DeclareAcronym{IPL}{short=\textsc{IPL}, long = Interferometric Phase Linking}
\DeclareAcronym{COFI-PL}{short=\textsc{COFI-PL}, long = COvariance Fitting Interferometric Phase Linking}
\DeclareAcronym{MM}{short=\textsc{MM}, long = Majorization-Minimization}
\DeclareAcronym{S-COFI-PL}{short=\textsc{S-COFI-PL}, long = Sequential COvariance Fitting Interferometric Phase Linking}
\DeclareAcronym{MSE}{short=\textsc{MSE}, long = Mean Squared Error}
\DeclareAcronym{RMSE}{short=\textsc{RMSE}, long = Root Mean Squared Error}
\DeclareAcronym{RGD}{short=\textsc{RGD}, long = Remannian Gradient Descent}
\DeclareAcronym{COMET}{short=\textsc{COMET}, long = COvariance Matching Estimation Technique}
\DeclareAcronym{FN}{short=\textsc{FN}, long = Frobenius Norm}
\DeclareAcronym{CES}{short=\textsc{CES}, long = Complex Elliptically Symmetric}
\DeclareAcronym{MT-InSAR}{short=MT-InSAR, long = Multi-Temporal Interferometric Synthetic Aperture Radar}
\DeclareAcronym{PSI}{short=\textsc{PSI}, long = Permanent Scatterer Interometry}
\DeclareAcronym{DSI}{short=\textsc{DSI}, long = Distributed Scatterer Interferometry }
\DeclareAcronym{PS}{short=\textsc{PS}, long = Permanent Scatterer}
\DeclareAcronym{DS}{short=\textsc{DS}, long = Distributed Scatterer}
\DeclareAcronym{SBAS}{short=\textsc{SBAS}, long = Small BAseline Subset}
\DeclareAcronym{ILS}{short=\textsc{ILS}, long = Integer Least Square}
\DeclareAcronym{EMI}{short=\textsc{EMI}, long = Eigendecomposition based Maximum likelihood estimator of Interferometric phase}
\DeclareAcronym{EVD}{short=\textsc{EVD}, long = EigenValue Decomposition}
\DeclareAcronym{PCA}{short=\textsc{PCA}, long = Principal Component Analysis}
\DeclareAcronym{NRT}{short=\textsc{NRT}, long = Near Real Time}
\DeclareAcronym{MLEPL}{short=\textsc{MLE-PL}, long = Phase Linking based on Maximum Likelihood Estimation}
\DeclareAcronym{SMLEPL}{short=\textsc{S-MLE-PL}, long = Sequential Phase Linking based on Maximum Likelihood Estimation}
\DeclareAcronym{i.i.d}{short=i.i.d, long = independent and identically distributed}
\DeclareAcronym{SVD}{short=\textsc{SVD}, long = Singular Value Decomposition}
\DeclareAcronym{BCD}{short=\textsc{BCD}, long = Block Coordinate Descent }

\DeclareAcronym{SK-PO}{short=\textsc{SK-PO}, long = Shrinking PO to identity}
\DeclareAcronym{SK-SCM}{short=\textsc{SK-SCM}, long = Shrinking SCM to identity}
\DeclareAcronym{BW-PO}{short=\textsc{BW-PO}, long = PO tapering}
\DeclareAcronym{BW-SCM}{short=\textsc{BW-SCM}, long = SCM tapering}

\DeclareAcronym{ERGAS}{short=\textsc{ERGAS}, long =  Relative Adimensional Global Error in Synthesis}
\DeclareAcronym{SAM}{short=\textsc{SAM}, long= Spectral Angle Mapper}
\DeclareAcronym{SCC}{short=\textsc{SCC}, long = Spatial Correlation Coefficient}
\DeclareAcronym{UQI}{short=\textsc{UQI}, long = Universal image Quality Index}
\DeclareAcronym{SSIM}{short=\textsc{SSIM}, long = Structural SIMilarity index}

%% file: accessories/title.tex
\title{Sequential Covariance Fitting \\
    for InSAR Phase Linking}
\author{Dana~EL~HAJJAR,~\IEEEmembership{Student~Member,~IEEE,}
        Guillaume~GINOLHAC,~\IEEEmembership{Senior~Member,~IEEE,}
        Yajing~YAN,~\IEEEmembership{Member,~IEEE,}
        and~Mohammed~Nabil~EL~KORSO 
\thanks{D. EL HAJJAR is with LISTIC (EA3703), University Savoie Mont-Blanc, and with L2S, CentraleSupélec, University Paris-Saclay. G. GINOLHAC and Y. YAN are with LISTIC (EA3703), University Savoie Mont-Blanc. M.N. EL KORSO is with L2S, CentraleSupélec, University Paris-Saclay. The source code is available on GitHub at the following address: \href{https://github.com/DanaElhajjar/S-COFI-PL}{https://github.com/DanaElhajjar/S-COFI-PL}}}

\markboth{Journal of \LaTeX\ Class Files} 
{Shell \MakeLowercase{\textit{et al.}}: Bare Demo of IEEEtran.cls for Journals}

\maketitle

%% file: Sections/abstract.tex
\begin{abstract}
    Traditional \ac{PL} algorithms are known for their high cost, especially with the huge volume of \ac{SAR} images generated by Sentinel-$1$ \ac{SAR} missions. Recently, a \ac{COFI-PL} approach has been proposed, which can be seen as a generic framework for existing \ac{PL} methods. Although this method is less computationally expensive than traditional PL approaches, \ac{COFI-PL} exploits the entire covariance matrix, which poses a challenge with the increasing time series of \ac{SAR} images. However, \ac{COFI-PL}, like traditional \ac{PL} approaches, cannot accommodate the efficient inclusion of newly acquired \ac{SAR} images. This paper overcomes this drawback by introducing a sequential integration of a block of newly acquired \ac{SAR} images. Specifically, we propose a method for effectively addressing optimization problems associated with phase-only complex vectors on the torus based on the \acl{MM} framework.
\end{abstract}

\IEEEpeerreviewmaketitle

%% file: Sections/notations.tex
\section*{Notations}
    \begin{table}[htbp]
        \centering
        \begin{tabularx}{\columnwidth}{l X}
            \textbf{Notations} & \textbf{Description} \\
            $\circ$ & element-wise (Hadamard) product \\
            $.^H$ & transposed and complex conjugated operator (Hermitian) \\
            $.^T$ & transposed operator \\
            $\widetilde{.}$ & tilde accentuation representing the entire dataset \\
            $\Bar{.}$ & bar accentuation representing the new items \\
            $.^*$ & complex conjugated operator \\
            $\textit{Re}(.)$ & complex number real component \\
            $\text{tr}(.)$ & trace matrix operator \\
            $\text{log}(.)$ & logarithm operator \\
            $\mathbb{E}(.)$ & first-order moment operator \\
            $||.||^2_F$ & frobenius norm \\
            $|.|$ & modulus operator \\
        \end{tabularx}
    \end{table}

%% file: Sections/introduction.tex
\section{Introduction}
     \IEEEPARstart{I}{\MakeLowercase{n}SAR} techniques have seen a great development and improvement with the immense number of available \acl{SAR} (\ac{SAR}) images \cite{osmanouglu2016time}, thanks to ongoing remote sensing missions such as Sentinel$-1$, Terra\ac{SAR}-X, etc.  \acl{InSAR} (\acs{InSAR}) is widely used to monitor land movements for various applications such as volcano monitoring, natural disaster damages, urban development. The first \acs{InSAR} techniques involved using two \acs{SAR} images acquired at two different dates ($2$-pass \acs{InSAR}) to estimate ground movements. However, when the dates are far apart, the coherence between the two images 
     
    \section*{Acronyms}
    \begin{table}[!h]
        \centering
        \begin{tabularx}{\columnwidth}{l X}
            \acs{CCG} & \acl{CCG} \\
            \acs{COFI-PL} & \acl{COFI-PL} \\
            \acs{COMET} & \acl{COMET} \\
            \acs{DS} & \acl{DS} \\
            \acs{DSI} & \acl{DSI} \\
            \acs{EVD} & \acl{EVD} \\
            \acs{ILS} & \acl{ILS} \\
             \acs{InSAR} & \acl{InSAR} \\
            \acs{KL} & \acl{KL} \\
            \acs{MLE} & \acl{MLE} \\
            \acs{MLEPL} & \acl{MLEPL} \\
            \acs{MM} & \acl{MM} \\
            \acs{MSE} & \acl{MSE} \\
            \acs{MT-InSAR} & \acl{MT-InSAR} \\
            \acs{PCA} & \acl{PCA} \\
            \acs{PL} & \acl{PL} \\
            \acs{PO} & \acl{PO} \\
            \acs{PS} & \acl{PS} \\
            \acs{PSI} & \acl{PSI} \\
            \acs{RGD} & \acl{RGD} \\
            \acs{SAR} & \acl{SAR} \\
            \acs{SBAS} & \acl{SBAS} \\
            \acs{SCM} & \acl{SCM} \\
            \acs{S-COFI-PL} & \acl{S-COFI-PL} \\
            \acs{SMLEPL} & \acl{SMLEPL} \\
        \end{tabularx}
    \end{table}
    
    \noindent weakens, resulting in a significant information loss. As a result, leveraging the vast availability of \ac{SAR} images from various missions, time-series-based approaches have been proposed to exploit the temporal correlation among these images. This concept can be summarized as \acl{MT-InSAR} (\acs{MT-InSAR}). 
    An essential family in \ac{MT-InSAR} analysis is called \acl{DSI} (\ac{DSI}) which exploits groups of homogeneous scatterers called \acl{DS} (\acs{DS}). A famous approach in \ac{DSI}, called \acl{PL} (\acs{PL}) \cite{ferretti2011new, ansari2018efficient, guarnieri2008exploitation, vu2022new, 9763551}, consists to retrieve \ac{SAR} image phases from the full \ac{SAR} images time series. 
    For an overview of existing \ac{PL} algorithms, refer to \cite{10261889, cao2015mathematical}.
    The general principle consists of two steps: first, estimating the \acl{SCM} (\ac{SCM}) and using its modulus as a plug-in of the unknown coherence matrix, and second, estimating the phases vector of \ac{SAR} images. However, this plug-in is non optimal due to the fact that it does not represent the \acl{MLE} (\ac{MLE}) of the coherence matrix, since the modulus operator is not holomorphic.
    The second approach, introduced in \cite{vu2022new, 9763551, vu2023robust}, presents the true \ac{MLE} of the coherence matrix jointly with the phases vector and will be referenced as \ac{MLEPL}, for \acl{MLEPL}, in the remaining of the paper.
    
    \noindent
    For both approaches, the process is based on the optimization of the negative log-likelihood function of the data following a zero mean \acl{CCG} (\acs{CCG}) distribution, which may not always fit the data.
    Studies such as \cite{mian2018new, vu2023robust} have shown that \ac{SAR} images are better modeled by a heavy-tailed distribution than a Gaussian distribution.
    Some works developed versions of the \ac{PL} based on a Non-Gaussian distribution \cite{vu2023robust} which demonstrated an improvement in the results. 
    These  algorithms are computationally intensive as they rely on the full covariance matrix estimation and consist of an iterative algorithm with two steps \cite{vu2022new, vu2023robust}: \textit{i)} estimating the phases vector, with an iterative \acl{MM} (\ac{MM}) algorithm, and \textit{ii)} estimating the coherence matrix. These two steps are also part of an iterative algorithm called \acl{BCD} (\ac{BCD}). Additionally, in the non-Gaussian case, the estimation of the covariance matrix is computationally intensive, as it involves Tyler-type estimators \cite{vu2023robust}, for which the estimation process itself is recursive.
    
    \noindent
    From above, it is clear that the estimation of the covariance matrix plays an important role. 
    An alternative of the \ac{MLE} approach was proposed in \cite{ottersten1998covariance}, called \acl{COMET} (\ac{COMET}), often requiring less computational effort while maintaining similar asymptotic properties. They consist of a set of statistical methods used for parameter estimation by matching the empirical and theoretical covariance matrices. They gained attention in the signal processing field due to their efficiency in solving various estimation problems. Covariance fitting approach was used in \acs{InSAR} in \cite{vu2024covariance}, namely \acl{COFI-PL} (\acs{COFI-PL}). 
    The main idea is to refine the structure of a covariance matrix estimator by minimizing a projection criterion. \ac{COFI-PL} can be seen as a generalization of most of the existing \ac{PL} problems \cite{vu2024covariance} while maintaining a lower cost than traditional \ac{PL} methods. The problem centers on two key components: the chosen method for covariance matrix estimation and the selected matrix distance. These factors can significantly influence the overall performance. Various covariance matrix plug-in can be used such as the \ac{SCM}, the \acl{PO} (\ac{PO}), Tyler estimator, etc.  \, Additionally, regularization techniques can also be adopted.  Although theoretically less effective, \acs{COFI-PL} has demonstrated strong and promising results in practice with real-world data \cite{vu2024covariance}. A key advantage of this approach is its robustness, which stems from to the choice of the covariance matrix plug-in. 

    \noindent
    However, when a new image or several new images are added, traditional methods become computationally intensive because they depend on the estimation of the full covariance matrix. This matrix increases with each new \ac{SAR} image, necessitating a complete re-estimation every time a new image is introduced. As a result, traditional methods struggle to manage this growth, imposing significant constraints on real-time processing of \acs{SAR} image time series. This often leads to inefficiencies and may fall short of operational requirements.
    
    \noindent
    To the best of our knowledge, only few studies addressed the aforementioned issue. Specifically, a sequential \ac{PL} approach was proposed in \cite{elhajjar2024}, namely, \acl{SMLEPL} (\acs{SMLEPL}), which incorporates each newly acquired \ac{SAR} image with a much lower computational cost while keeping the same performances as \ac{MLEPL}. Despite the good performance of this approach and the reduction in computation time, it incorporates new \ac{SAR} images one by one within a Gaussian model framework. 
    A sequential approach, proposed in \cite{ansari2017sequential}, is able to incorporate a bloc of new images. It is based on treating each mini-stack by applying the classic \ac{PL} and then compressing it into a single virtual image through \acl{PCA} (\ac{PCA}). Each obtained virtual image is then connected to the next mini-stack. This method uses a partial coherence matrix, which can limit the amount of information available for analysis and reduce the accuracy of the estimates.

    In this paper, we aim to develop a sequential approach that incorporates a stack of new \ac{SAR} images, based on \ac{COFI-PL} framework. We employ two matrix distances : \acl{KL} (\acs{KL}) divergence and the Frobenius Norm, along with various covariance matrix estimators and potential regularization techniques. 
    For each problem, we propose a \ac{MM} algorithm to solve it. Results from simulations and real data demonstrate comparable outcomes to the benchmark, namely the offline approach \ac{COFI-PL}.
    
    In the following, a brief overview of \ac{COFI-PL} is proposed in Section \ref{section:cov_fitting_offline}.  We present two matrix distances in Section \ref{subsec:mat_dist}, different estimators of the covariance matrix in Section \ref{subsec:plug_in} and various ways to solve each problem in Section \ref{subsec:algo}. The description of our contribution starts at section \ref{section:seq_cov_fitting}.
    Validation with experiments and real data are presented in Sections \ref{section:exp} and \ref{section:real_data}.

%% file: Sections/generic_cov_fitting.tex
\section{Generic covariance fitting problem}
\label{section:cov_fitting_offline}
Covariance fitting problems are prevalent across various domains, and several works have proposed different versions to address them \cite{meriaux2019recursions, meriaux2019robust, hu2010psf, werner2008estimation, meriaux2017robust, meriaux2020matched, vu2024covariance}.
\acs{COMET} consists of two steps: first, to estimate a plug-in of the covariance matrix, and second, to use a fitting criterion between the computed plug-in and a particular structure of the covariance matrix.
 In \ac{InSAR}, \cite{bai2023lamie, vu2024covariance} proposed covariance fitting approaches to estimate phases while respecting the phase closure property.

\subsection{\ac{InSAR} model and optimization problem}
For a given stack of $l$ \ac{SAR} images, we denote $\{\mathbf{\widetilde{x}}^i\}_{i=1}^n$ the local homogeneous spatial neighborhood of size $n$ for each pixel, where $\mathbf{\widetilde{x}}^i \in \mathbb{C}^{l}$, for all $i \in [\![1,n]\!]$. 
To ensure the temporal phase closure property, and based on standard physical properties \cite{guarnieri2008exploitation}, the covariance matrix of $\{\mathbf{\widetilde{x}}^i\}_{i=1}^n$ can be formulated as 
\begin{equation}
   \mathbf{\widetilde{\Sigma}} = \mathbf{\widetilde{\Psi}} \odot \mathbf{\widetilde{w}}_{\mathbf{\theta}}\mathbf{\widetilde{w}}_{\mathbf{\theta}}^H
   \label{cov_mat_struc}
\end{equation}
where $\mathbf{\widetilde{\Psi}}$ is the real core of the covariance matrix ($\mathbf{\widetilde{\Psi}} = |\mathbf{\widetilde{\Sigma}}|$), i.e. the coherence matrix scaled by the variances coefficients, $\mathbf{\widetilde{w}}_{\mathbf{\theta}}$ denotes the vector of the exponential of the complex phases. %

\noindent
In \ac{InSAR}, the goal is to estimate the phases vector $\mathbf{\widetilde{w}}_{\mathbf{\theta}}$. Using the structure of the covariance matrix in equation (\ref{cov_mat_struc}), the vector $\mathbf{\widetilde{w}}_{\mathbf{\theta}}$ can be estimated with \ac{COMET}. At this point, the optimization problem is represented as 
\begin{equation}
    \begin{array}{c l}
    \underset{\mathbf{\widetilde{w}}_{\theta}}{\rm minimize}
    & 
    f^d_{\mathbf{\widetilde{\Sigma}}} (\mathbf{\widetilde{\Sigma}}, \mathbf{\widetilde{\Psi}} \circ \mathbf{\widetilde{w}}_{\theta} \mathbf{\widetilde{w}}_{\theta}^H)
        \\
        {\text{ subject~to} }
    & \theta_1 = 0
    \\
     & \mathbf{\widetilde{w}}_{\theta} \in \mathbb{T}_l \\
     
    \end{array}
    \label{eq:prob_optim_fitting_offline}
\end{equation}
with $\mathbb{T}_l = \{  \mathbf{\widetilde{w}}_{\theta} \in \mathbb{C}^l  \, | \, |[\mathbf{\widetilde{w}}_{\theta}]_i| = 1, \forall i \in [1, l] \}$. 
The objective function of the problem (\ref{eq:prob_optim_fitting_offline}) is a matrix distance between the plug-in of the covariance matrix (possibly regularized) and its projection where the subscript $\mathbf{\widetilde{\Sigma}}$ consists of the covariance matrix plug-in, the superscript $\textit{d}$ denotes the chosen distance. 
Based on \cite{vu2024covariance}, we consider two distances, the \acl{KL} (\ac{KL}) divergence and the Frobenius norm. For the plug-in, the classic \ac{SCM} as well as a robust version will be used. For regularization, we consider the shrinking to identity and the tapering. The primary objective is to demonstrate the sequential adaptation of the approach rather than evaluating the efficiency of each distance and each estimator. The corresponding distances are presented in the following section.

\subsection{Matrix distances}
\label{subsec:mat_dist}
In this section, we introduce the \ac{KL} divergence and the Frobenius norm. 

\vspace{0.5cm}

\subsubsection{\texorpdfstring{\acl{KL} divergence}{KL divergence}}
\hfill

The \ac{KL} divergence 
measures the similarity between two probability density functions. 
Between two centered Gaussian distributions, the \ac{KL} divergence has the following form
\begin{equation}
    \begin{aligned}
        KL(\mathcal{CN}(0, \mathbf{\Sigma}_1) & \parallel \mathcal{CN}(0, \mathbf{\Sigma}_2)) \\
        & = \text{tr}(\mathbf{\Sigma}_2^{-1} \mathbf{\Sigma}_1) + \log(|\mathbf{\Sigma}_2 \mathbf{\Sigma}_1^{-1}|) - l
    \end{aligned}
\end{equation}

\noindent
In our context, $\mathbf{\Sigma}_1$ represents the plug-in of the covariance matrix, and $\mathbf{\Sigma}_2$ corresponds to the covariance matrix structure in equation (\ref{cov_mat_struc}), leading to the following structure of the objective function
\begin{equation}
\label{cost_fct_KL}
    f^{\text{KL}}_{\mathbf{\widetilde{\Sigma}}}(\mathbf{\widetilde{w}}_{\theta}) = \mathbf{\widetilde{w}}_{\theta}^H (\mathbf{\widetilde{\Psi}}^{-1} \circ \mathbf{\widetilde{\Sigma}}) \mathbf{\widetilde{w}}_{\theta} \\
\end{equation}
The \ac{KL} divergence, shown in equation (\ref{cost_fct_KL}), between the Gaussian distribution with \ac{SCM} as a plug-in for the covariance matrix and the Gaussian distribution with the structure of the covariance matrix in equation (\ref{cov_mat_struc}), corresponds to the objective function of the classic \ac{PL} \cite{guarnieri2008exploitation}.

\vspace{0.5cm}

\subsubsection{Frobenius norm}
\hfill

The Frobenius norm (also called Euclidean distance) of a matrix is defined as the sum of the absolute squares of its elements. The Frobenius norm between two symmetric matrices is defined as
\begin{equation}
    d^2_F(\mathbf{\Sigma}_1, \mathbf{\Sigma}_2) = || \mathbf{\Sigma}_1 - \mathbf{\Sigma}_2 ||^2_F
\end{equation}
By setting $\mathbf{\Sigma}_1 = \mathbf{\widetilde{\Sigma}}$ and $\mathbf{\Sigma}_2 = \mathbf{\widetilde{\Psi}} \odot \mathbf{\widetilde{w}}_{\theta}\mathbf{\widetilde{w}}_{\theta}^H$, the simplified form of the corresponding objective function is obtained as
\begin{equation}
\label{cost_fct_LS}
    f^{\text{FN}}_{\mathbf{\widetilde{\Sigma}}}(\mathbf{\widetilde{w}}_{\theta}) = - 2 \mathbf{\widetilde{w}}_{\theta}^H (\mathbf{\widetilde{\Psi}} \circ \mathbf{\widetilde{\Sigma}}) \mathbf{\widetilde{w}}_{\theta} \\
\end{equation}
As it appears, the main interest between the two objective functions is the absence of the covariance matrix inversion in equation (\ref{cost_fct_LS}).

\subsection{Covariance matrix plug-in}
\label{subsec:plug_in}

In this section, we present the various covariance matrix plug-in (Section \ref{subsubsect:unstruct_plugin}) and the possible regularization (Section \ref{subsubsect:regul_plugin}).
\
\vspace{0.5cm}

\subsubsection{Unstructured covariance matrix}
\label{subsubsect:unstruct_plugin}
\hfill

In this section, we present two estimators: the \acl{SCM} (\ac{SCM}) and the \acl{PO} (\ac{PO}). In most \ac{PL} approaches, the \ac{SCM} is used and yields efficient results. In \cite{vu2024covariance}, the \ac{PO} estimator was proposed and showed an improvement compared to other plug-in. 
In the remainder of the paper, the superscript index of a matrix represents the choice of the plug-in method ($U$ : unconstrained, $SK$ : shrinkage to identity regularization and $BW$ : bandwidth (tapering regularization)).

\hfill

\paragraph{\acl{SCM} (\ac{SCM})}

\ac{SCM} is a common plug-in for the covariance matrix in \ac{InSAR} processing \cite{guarnieri2008exploitation, ferretti2001permanent, fornaro2014caesar, ferretti2011new, cao2015mathematical}. This choice comes from the \ac{MLE} of the covariance matrix $\mathbf{\widetilde{\Sigma}}$ according to a centered \ac{CCG} model ($\mathbf{\widetilde{x}} \sim \mathcal{CN}(0, \mathbf{\widetilde{\Sigma}})$)
\begin{equation}
    \mathbf{\widetilde{\Sigma}}^{U} = \frac{1}{n} \sum_{i=1}^n \mathbf{\widetilde{x}}^i \mathbf{\widetilde{x}}^{iH} 
    \label{eq:SCM}
\end{equation}
\ac{SCM} plug-in for the covariance matrix yields good results, however it is not robust when the data deviate from Gaussianity.
\ac{SCM} plug-in for the covariance matrix constitutes an efficient estimator when the data follows a Gaussian distribution. When the data deviate from Gaussianity or present outliers, a major deterioration in performances is observed  \cite{vu2023robust, vu2024covariance}.

\hfill

\paragraph{\acl{PO} (\ac{PO})}

Let $\Phi_{\mathbb{T}}: x = r e^{i \theta} \xrightarrow \; e^{i \theta}$ and $\mathbf{\widetilde{y}} = \Phi_{\mathbb{T}}(\mathbf{\widetilde{x}})$ represents the vector obtained by applying the transformation $\Phi_{\mathbb{T}}$ to the vector $\mathbf{\widetilde{x}}$. The \ac{PO} plug-in is defined as follow
\begin{equation}
    \mathbf{\widetilde{\Sigma}}^{U} = \frac{1}{n} \sum_{i=1}^n \mathbf{\widetilde{y}}^i \mathbf{\widetilde{y}}^{iH} 
    \label{eq:PO}
\end{equation}
The advantage of the corresponding estimator is the mitigation of amplitudes variations among \ac{SAR} images \cite{cao2015mathematical}, unlike the \ac{SCM}. 
Additionally, it offers a balance between robustness and computation time, unlike the Tyler estimator, which is significantly more costly to compute.

\vspace{0.5cm}

\subsubsection{Covariance matrix regularization}
\label{subsubsect:regul_plugin}
\hfill

\ac{PL} faces several challenges including the quality of the covariance matrix estimation. In fact, the accuracy of the covariance matrix estimation depends on the sample size $n$, which must be more than twice the length of the time series. For long time series, in situations such as $n \approx l$ or $n < l$, the covariance matrix estimation is deemed inaccurate. Additionally, in such cases, the matrix inversion may also pose challenges due to the matrix being either ill-conditioned or singular. Moreover, matrix inversion is a required step in all the algorithms \cite{guarnieri2008exploitation, vu2022new, vu2023robust, vu2024covariance}.
A popular strategy to improve the covariance matrix estimation, is to use a regularization form \cite{ledoit2004well, chen2010shrinkage, rucci2010skp, chen2012shrinkage, pascal2014generalized, ollila2022regularized}. 
The introduction of this regularization in \ac{IPL} was discussed in \cite{even2018insar}, and several works proposed such regularization \cite{zwieback2022cheap, bai2023lamie, vu2024covariance, liang2024coherence, zhao2024regularized}

\hfill

\paragraph{Shrinkage to identity}

As known, most covariance matrix estimators are unsatisfactory for large covariance estimation problems since they require huge sample support. In standard \ac{PL} approaches, there is a significant chance of needing to invert a matrix such as \ac{SCM}. This requires a large sample size to obtain an accurate estimation of the inversion. Additionally, the sample size increases with the length of the time series. The most common strategy for dealing with the problem of few samples is to shrink the estimator in question to a target matrix. An attractive case is to shrink the estimator to a scaled identity matrix \cite{chen2010shrinkage, even2018insar}. In other words, it is a combination between the covariance matrix estimator and the target matrix. The shrinkage parameter $\beta$ has a closed form under an assumption of Gaussian model \cite{chen2012shrinkage}.
\begin{equation}
    \mathbf{\widetilde{\Sigma}}^{SK} = \beta \mathbf{\widetilde{\Sigma}}^U + (1 - \beta) \frac{ \text{tr}(\mathbf{\widetilde{\Sigma}}^U)}{l} \mathbf{I}_l 
    \label{eq:cov_mat_shrink}
\end{equation}
with  $\beta \in [0, 1]$. Several combinations can be formed as in \cite{chen2012shrinkage, chen2011robust, ollila2022regularized}.

\hfill

\paragraph{Covariance matrix tapering}

\ac{InSAR} approaches based on \ac{PL} exploit all possible combinations between \ac{SAR} images through the use of the covariance matrix. However, long \ac{SAR} images time series suffer from target decorrelations as time progresses, resulting in near zero coherence among image pairs. 
The idea of the covariance matrix tapering regularization is to exclude pairs of images suffering from low coherence, which is done through the bandwidth $b$ of the tapering matrix. Several techniques can be used to choose the value of $b$ \cite{bickel2008regularized, bickel2008covariance, ollila2022regularized}. The interest of such regularization was discussed in \cite{bai2023lamie, vu2024covariance, liang2024coherence, zhao2024regularized}. 
The covariance matrix tapering has the following form 
\begin{equation}
    [\mathbf{\widetilde{W}}(b)]_{ij} = 
    \begin{cases}
    1 & \text{if} \quad |i - j| \leq b \\
    0 & \text{otherwise}
    \end{cases}
    \label{eq:tap}
\end{equation}
Using (\ref{eq:tap}), the covariance matrix tapering has the following form:
\begin{equation}
    \mathbf{\widetilde{\Sigma}}^{BW} = \mathbf{\widetilde{W}}(b) \circ \mathbf{\widetilde{\Sigma}}^U
\end{equation}

\noindent
A combination of shrinkage regularization and tapering was proposed in \cite{ollila2022regularized}. 
\subsection{Optimization method} 
\label{subsec:algo}
We propose a \acs{MM} algorithm to solve the optimization problem in (\ref{eq:prob_optim_fitting_offline}). We note that this problem can also be solved using \acl{RGD} (\ac{RGD}), which can be advantageous for cost functions that cannot be expressed in a quadratic form \cite{vu2024covariance}, making it unsuitable for the \ac{MM} algorithm. The benefit of Riemannian optimization on the manifold lies in its ability to automatically compute the euclidean gradient of any cost function, facilitated by Python's JAX package. However, this method requires identifying a hyper-parameter that will serve as the step size in the gradient descent. Each of these two methods has its advantages, but importantly, both ultimately lead to similar results \cite{vu2024covariance}, highlighting their effectiveness in solving the optimization problem. That said, the \ac{MM} algorithm has the additional advantage of being computationally faster than \ac{RGD} \cite{vu2024covariance}. Given that we use the \acs{KL} divergence and the Frobenius norm (with their ability to take a quadratic form), the \ac{MM} algorithm is sufficient.

\ac{MM} framework is composed of two main steps. The first step is called "Majorization", where the objective is to find a surrogate function $g(.|\mathbf{w}_{\theta}^{(t)})$ that majorizes the function $f(\mathbf{w}_{\theta})$ and whose form depends on $\mathbf{w}^{(t)}$
\begin{equation}
    f(\mathbf{w}_{\theta}) \leq g(\mathbf{w}_{\theta}|\mathbf{w}_{\theta}^{(t)}), \quad \forall \, \mathbf{w}_{\theta} \in \mathbb{T}_l
\end{equation}
In other words, $g(\mathbf{w}_{\theta}|\mathbf{w}_{\theta}^{(t)})$ is tangent to $f(\mathbf{w}_{\theta})$ at the point $\mathbf{w}_{\theta} = \mathbf{w}_{\theta}^{(t)}$.
For the \ac{KL} divergence and the Frobenius norm, the following Lemmas are used respectively to find the corresponding surrogate.
\begin{lemma}
    The convex quadratic form 
    \begin{equation}
        f : \mathbf{w}_{\theta} \longrightarrow \mathbf{w}_{\theta}^H \mathbf{H} \mathbf{w}_{\theta}
    \end{equation}
    is majored on $\mathbb{T}_l$ by : 
    \begin{equation}
        g(\mathbf{w}_{\theta}|\mathbf{w}_{\theta}^{(t)}) = 2 \textit{Re} \left( \mathbf{w}_{\theta}^H (\mathbf{H} - \lambda^{\mathbf{H}}_{max} \mathbf{I}) \mathbf{w}_{\theta}^{(t)} \right) + \text{const}
    \end{equation}
    with equality at point $\mathbf{w}_{\theta}^{(t)}$. $\lambda^{\mathbf{H}}_{max}$ corresponds to the largest eigenvalue of $\mathbf{H}$ \cite{vu2024covariance}.
    \label{lemma1}
\end{lemma}   
    
\begin{lemma}
    The concave quadratic form 
    \begin{equation}
        f : \mathbf{w}_{\theta} \longrightarrow - \mathbf{w}_{\theta}^H \mathbf{H} \mathbf{w}_{\theta}
    \end{equation}
    is majored, with equality at point $\mathbf{w}^{(t)}$ \cite{vu2024covariance},  by : 
    \begin{equation}
        g(\mathbf{w}_{\theta}|\mathbf{w}_{\theta}^{(t)}) = -2 \textit{Re}( \mathbf{w}^H \mathbf{H} \mathbf{w}_{\theta}^{(t)} )
    \end{equation}
    \label{lemma3}
\end{lemma}
    
\noindent
The second step, "Minimization" includes obtaining the next iterate of the corresponding parameter by minimizing the obtained function $g(.|\mathbf{w}_{\theta}^{(t)})$ rather than the actual function $f(\mathbf{w}_{\theta})$
\begin{equation}
    \mathbf{w}_{\theta}^{(t+1)} = \text{argmin}_{\mathbf{w}_{\theta} \in \mathbb{T}_l} \, g(\mathbf{w}_{\theta}|\mathbf{w}_{\theta}^{(t)})
\end{equation}

\noindent
For both matrix distances, Lemma \ref{lemma2} is used to minimize the surrogate function. 
\begin{lemma}
    The solution of the following minimization problem
    \begin{equation}
        \begin{array}{c l}
            \underset{\mathbf{\Bar{w}}_{\theta} \in \mathbb{T}_k}{\rm minimize}
            & 
            - \textit{Re}(\mathbf{\Bar{w}}_{\theta}^H \mathbf{\check{w}}_{\theta}^{(t)})
        \end{array}
    \end{equation}
    is obtained as $\mathbf{\Bar{w}}_{\theta}^* = \Phi_{\mathbb{T}}(\mathbf{\check{w}}_{\theta}^{(t)})$
    with $\Phi_{\mathbb{T}}: x = r e^{i \theta} \xrightarrow \; e^{i \theta}$ \cite{vu2024covariance}
    \label{lemma2}
\end{lemma}
    
\noindent
The main key to a successful \ac{MM} algorithm consists in the step of constructing the function $g(.|\mathbf{w}_{\theta}^{(t)})$ \cite{sun2016majorization}. The advantage of this framework remains in the low computational cost of the obtained optimization problem. \ac{MM} algorithms have been already used in \ac{InSAR} \cite{vu2022new, vu2023robust}. Further details on \ac{MM} algorithms are provided in \cite{hunter2004tutorial, razaviyayn2013unified, soltanalian2014designing, sun2016majorization, breloy2021majorization}.

\noindent
To summarize this section, we defined the optimization problem for estimating the complete  phase time series, as well as the matrix distances and covariance matrix plug-in. In the following section, we will present the sequential  formulation of the \ac{COFI-PL} problem.

%% file: Sections/sequential_cov_fitting.tex
\section{Sequential formulation of covariance fitting}
\label{section:seq_cov_fitting}

\subsection{Data model and optimization problem}

We consider a stack of $l = p+k$ \ac{SAR} images ($p$ represents the length of the \ac{SAR} images time series of the past and $k$ the number of newly acquired \ac{SAR} images). $\{\mathbf{\widetilde{x}}^i\}_{i=1}^n$ denotes the local homogeneous spatial neighborhood of size $n$ for a given pixel, where $\mathbf{\widetilde{x}}^i \in \mathbb{C}^{l}$, for all $i \in [\![1,n]\!]$, i.e., 
\begin{equation}
    \mathbf{\widetilde{x}}^i = [\underbrace{x_1^i, \dots, x_p^i}_{\mathbf{x}^i}, \underbrace{x_{p+1}^i, \dots,  x_{l}^i}_{\mathbf{\Bar{x}}^i}]^T \in \mathbb{C}^{l}
    \label{eq:data_model}
\end{equation}
where $\mathbf{x}^i \in \mathbb{C}^{p}$ denotes the multivariate pixel of the previous data, and  $\mathbf{\Bar{x}}^i \in \mathbb{C}^{k}$ for the new ones (Fig. \ref{fig:new_bloc}). Each pixel of the local patch is assumed to be distributed as a zero mean \acf{CCG} \cite{bamler1998synthetic}, i.e., $\mathbf{\widetilde{x}}\sim\mathcal{CN}(0, \mathbf{\widetilde{\Sigma}})$.

\begin{figure}[hbt]
    \centering
    \input{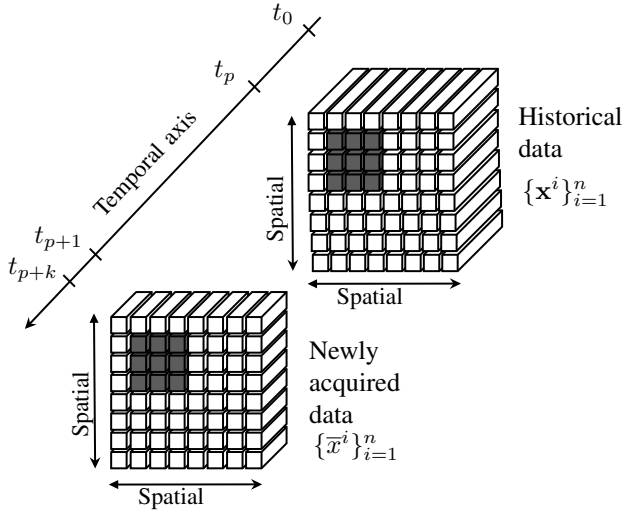}
    \caption{\small \ac{SAR} data representation including both previous and recently obtained block of images. The local neighborhood of size $n$ is denoted by gray pixels (sliding window).}
    \label{fig:new_bloc}
\end{figure}

\noindent
The vector of phases of the entire set of images can be expressed as a block structure as follow 
\begin{equation}
    \mathbf{\widetilde{w}}_{\theta} = \begin{pmatrix}
        \mathbf{w}_{\theta} \\
        \mathbf{\Bar{w}}_{\theta}
    \end{pmatrix}
    \label{eq:w_bloc}
\end{equation}

\noindent
where $\mathbf{w}_{\theta}$ denotes the exponential phases vector for the previous data, and $\mathbf{\Bar{w}}_{\theta}$ denotes the exponential phase vector for the new ones.
The covariance matrix of the entire set of images can be expressed as the following block structure
\begin{align}
    \mathbf{\widetilde{\Sigma}}
    &= \left( \begin{array}{c}
        \begin{tabular}{cc}
           $\mathbf{\Sigma}_{p}$ & $(\mathbf{\Sigma}_{pn})^H$ \\
                $\mathbf{\Sigma}_{pn}$ &  $\mathbf{\Sigma}_{n}$ \\
        \end{tabular}
    \end{array} \right)
\label{eq:bloc_cov}
\end{align}

\noindent
where $\mathbf{\Sigma}_{p}$ represents the covariance and coherence matrices, respectively, of the block of past images, both of which have dimensions $(p, p)$. Meanwhile, $\mathbf{\Sigma}_{pn}$ denotes the covariance and coherence, respectively, between the past images and the new ones with dimensions $(k, p)$. Additionally, $\mathbf{\Sigma}_{n}$ corresponds to the covariance and coherence matrices of the block of new images, with dimensions $(k, k)$.

\noindent
In the remainder of the paper, the subscript index of a matrix indicates the block type ($p$ : past images, $pn$ : interaction between past and new images and $n$ : new images). For each plug-in, we present its detailed block structure as follow

\hfill

\subsubsection{Unstructured plug-in: SCM or PO}
\label{subsubsec:unconstrained_seq}
The covariance matrix, whether it is \ac{SCM} (equation (\ref{eq:SCM})) or \ac{PO} (equation (\ref{eq:PO})) , can be represented in the following block structure
\begin{align}
    \mathbf{\widetilde{\Sigma}}^{U}
    &= \left( \begin{array}{c}
        \begin{tabular}{cc}
           $\mathbf{\Sigma}_{p}^{U}$ & $(\mathbf{\Sigma}_{pn}^{U})^H$ \\
                $\mathbf{\Sigma}_{pn}^{U}$ &  $\mathbf{\Sigma}_{n}^{U}$ \\
        \end{tabular}
    \end{array} \right)
\label{eq:bloc_cov_unstruct}
\end{align}

\hfill

\subsubsection{Shrinkage to identity}
\label{subsubsec:shrinkage_seq}
The block structure of the shrinkage to identity regularization of a plug-in of the covariance matrix in equation (\ref{eq:cov_mat_shrink}) can be represented as follow

\noindent\resizebox{\columnwidth}{!}{
\begin{minipage}{\columnwidth}
\begin{align}
    \mathbf{\widetilde{\Sigma}}^{SK} &= \beta \mathbf{\widetilde{\Sigma}}^{U} + (1 - \beta) \frac{ \text{tr}( \mathbf{\widetilde{\Sigma}}^{U})}{l} \mathbf{I}_l \notag \\
    &= \left( \begin{array}{cc}
        \beta \mathbf{\Sigma}_{p}^{U} + (1 - \beta) \frac{\text{tr}(\mathbf{\widetilde{\Sigma}}^{U})}{l} \mathbf{I}_p & \beta (\mathbf{\Sigma}_{pn}^{U})^H \\
        \beta \mathbf{\Sigma}_{pn}^{U} & \beta \mathbf{\Sigma}_{n}^{U} + (1 - \beta) \frac{\text{tr}(\mathbf{\widetilde{\Sigma}}^{U})}{l} \mathbf{I}_k \\
    \end{array} \right) \notag \\
    &= \left( \begin{array}{c}
        \begin{tabular}{cc}
           $\mathbf{\Sigma}_{p}^{SK}$ & $(\mathbf{\Sigma}_{pn}^{SK})^H$ \\
                $\mathbf{\Sigma}_{pn}^{SK}$ &  $\mathbf{\Sigma}_{n}^{SK}$ \\
        \end{tabular}
    \end{array} \right)
\label{eq:bloc_shrink}
\end{align}
\end{minipage}
}

\noindent
given that the estimator $\mathbf{\widetilde{\Sigma}}^{U}$ can be either the \ac{SCM} or \ac{PO} as in equation (\ref{eq:bloc_cov_unstruct}).
    
\hfill

\subsubsection{Covariance matrix tapering}
\label{subsubsec:tapering_seq}
The tapering matrix with bandwidth $\textit{b}$ in equation (\ref{eq:tap}) can have the following block structure
\begin{align}
    \label{eq:bloc_tapering}
    \mathbf{\widetilde{W}}(b)
    &= \left( \begin{array}{c}
        \begin{tabular}{ccc}
            $\mathbf{W}_{p}$ & $(\mathbf{W}_{pn})^T$ \\
            $\mathbf{W}_{pn}$ &  $\mathbf{W}_{n}$ \\
        \end{tabular}
    \end{array} \right)
\end{align}
\noindent
Using the structure in equation (\ref{eq:bloc_cov}), the tapered covariance matrix can be partitioned into a block structure comprising three distinct components: one representing the past $\mathbf{W}_{p} \circ \mathbf{\Sigma}_{p}$, one corresponding to the new data $ \mathbf{W}_{n} \circ \mathbf{\Sigma}_{n}$, and a third capturing the interaction between the two $\mathbf{W}_{pn} \circ \mathbf{\Sigma}_{pn}$ as follow
\begin{align}
    \mathbf{\widetilde{\Sigma}}^{BW} &=  \mathbf{\widetilde{W}}(b) \circ \mathbf{\widetilde{\Sigma}}^{U} \notag \\
    &= \left( \begin{array}{c}
        \begin{tabular}{cc}
            $\mathbf{W}_{p} \circ \mathbf{\Sigma}_{p}^{U}$ & $(\mathbf{W}_{pn})^T \circ (\mathbf{\Sigma}_{pn}^{U})^H$ \\
            $\mathbf{W}_{pn} \circ \mathbf{\Sigma}_{pn}^{U}$      &  $ \mathbf{W}_{n} \circ \mathbf{\Sigma}_{n}^{U}$ \\
        \end{tabular}
    \end{array} \right) \notag  \\ 
    &=  \left( \begin{array}{c}
        \begin{tabular}{cc}
           $\mathbf{\Sigma}_{p}^{BW}$ & $(\mathbf{\Sigma}_{pn}^{BW})^H$ \\
                $\mathbf{\Sigma}_{pn}^{BW}$ &  $\mathbf{\Sigma}_{n}^{BW}$ \\
        \end{tabular}
    \end{array} \right)
     \label{eq:tapered_bloc_cov}
\end{align}

\hfill

\noindent
The optimization problem is formulated as follow
\begin{equation}
    \begin{array}{c l}
    \underset{\mathbf{\Bar{w}}_{\theta}}{\rm minimize}
    & 
    f^d_{\mathbf{\widetilde{\Sigma}}} (\mathbf{\widetilde{\Sigma}}, \mathbf{\widetilde{\Psi}} \circ \mathbf{\widetilde{w}}_{\theta} \mathbf{\widetilde{w}}_{\theta}^H)
        \\
        {\text{ subject~to} }
    & \theta_1 = 0
    \\
     & \mathbf{\Bar{w}}_{\theta} \in \mathbb{T}_k \\
     
    \end{array}
    \label{eq:prob_optim_fitting_online}
\end{equation}
\noindent
where $\mathbf{\Bar{w}}_{\theta}$  is defined in equation (\ref{eq:w_bloc}).  $f^d_{\mathbf{\widetilde{\Sigma}}} (\mathbf{\widetilde{\Sigma}}, \mathbf{\widetilde{\Psi}} \circ \mathbf{\widetilde{w}}_{\theta} \mathbf{\widetilde{w}}_{\theta}^H)$ is the cost function of the \ac{COMET} problem and depends on $\mathbf{\Bar{w}}_{\theta}$. For each distance, a generic form is provided  in the sections (\ref{subsection:KL_seq}, \ref{subsection:Frob_seq}).

\begin{table}[!h]
\renewcommand{\arraystretch}{1.3} 
\centering
\begin{tabular}{|c|c|c|c|}
\hline
\multirow{2}{*}{} &  \makecell{unconstrained  \\ plug-in  \\ (section \ref{subsubsec:unconstrained_seq})} 
                  & \makecell{shrinkage  \\ to identity  \\ (section \ref{subsubsec:shrinkage_seq})} 
                  &  \makecell{tapering  \\ regularization  \\ (section \ref{subsubsec:tapering_seq})} \\ \cline{2-4}
                  &  \acs{SCM} $\setminus$ \acs{PO} & \acs{SK-SCM} $\setminus$ \acs{SK-PO} &  \acs{BW-SCM} $\setminus$ \acs{BW-PO} \\ \hline
$\mathbf{\Sigma}_{p}$ & $\mathbf{\Sigma}_{p}^{U}$ & $\mathbf{\Sigma}_{p}^{SK}$ &  $\mathbf{\Sigma}_{p}^{BW}$ \\ \hline
$\mathbf{\Sigma}_{pn}$         & $\mathbf{\Sigma}_{pn}^{U}$ & $\mathbf{\Sigma}_{pn}^{SK}$ & $\mathbf{\Sigma}_{pn}^{BW}$ \\ \hline
$\mathbf{\Sigma}_{n}$    & $\mathbf{\Sigma}_{n}^{U}$ & $\mathbf{\Sigma}_{n}^{SK}$ & $\mathbf{\Sigma}_{n}^{BW}$ \\ \hline
\end{tabular}
\caption{$\mathbf{\Sigma}_{p}$, $\mathbf{\Sigma}_{pn}$ and $\mathbf{\Sigma}_{n}$ values depending on the choice of the covariance matrix plug-in.}
\label{tab:tab_plug_in}
\end{table}

\subsection{\acl{KL} divergence}
\label{subsection:KL_seq}
In this section, we provide a general form of the cost function for the \ac{KL} divergence presented in Proposition \ref{prop:prop_CF_KL}, which can be applied and modified according to the choice of the covariance matrix plug-in in Table \ref{tab:tab_plug_in}
\begin{proposition} 
\textbf{Generic cost function for \ac{KL} divergence}

Regardless the choice of the covariance matrix plug-in, the cost function for the \ac{KL} divergence has the following form
    \begin{align}
    \label{eq:CF_KL}
        f^{\text{KL}}_{\mathbf{\widetilde{\Sigma}}}(\mathbf{\Bar{w}}_{\theta}) &= \mathbf{w}_{\theta}^H \left( \mathbf{F}^{-1} \circ \mathbf{\Sigma}_{p} \right) \mathbf{w}_{\theta}  \notag \\
        & + \mathbf{w}_{\theta}^H \left( ( \mathbf{A})^T \circ  (\mathbf{\Sigma}_{pn})^H \right)\mathbf{\Bar{w}}_{\theta} \notag \\
        & + \mathbf{\Bar{w}}_{\theta}^H \left( \mathbf{A} \circ  \mathbf{\Sigma}_{pn} \right) \mathbf{w}_{\theta} \notag \\
        & + \mathbf{\Bar{w}}_{\theta}^H \mathbf{M} \mathbf{\Bar{w}}_{\theta} 
    \end{align}
    where
\begin{itemize}
    \item[-] $\mathbf{F} = |\mathbf{\Sigma}_{p}| - (|\mathbf{\Sigma}_{pn}|)^T |\mathbf{\Sigma}_{n}|^{-1}|\mathbf{\Sigma}_{pn}|$
    \item[-] $\mathbf{D} = |\mathbf{\Sigma}_{n}| - |\mathbf{\Sigma}_{pn}||\mathbf{\Sigma}_{p}|^{-1}(|\mathbf{\Sigma}_{pn}|)^T$
    \item[-] $\mathbf{A} = - \mathbf{D}^{-1} |\mathbf{\Sigma}_{pn}||\mathbf{\Sigma}_{p}|^{-1}$
    \item[-] $\mathbf{M} = \mathbf{D}^{-1} \circ \mathbf{\Sigma}_{n}$
\end{itemize}
and 
$\mathbf{\Sigma}_{p}$, $\mathbf{\Sigma}_{pn}$ and $\mathbf{\Sigma}_{n}$ depend on the choice of the covariance matrix plug-in Table  \ref{tab:tab_plug_in}. 
    \label{prop:prop_CF_KL}
\end{proposition}
\begin{proof}
in Appendix \ref{app_CF_KL}.
\end{proof}

For each  case of the generic cost function for the \ac{KL} divergence in equation (\ref{eq:CF_KL}) and after a few calculations steps, it corresponds to the convex quadratic form in Lemma \ref{lemma1}. Similar to the cost function, we present a generic form for the surrogate function that majorizes the corresponding cost function in equation (\ref{eq:CF_KL}).
    
    \begin{proposition}
    \textbf{Surrogate function for \ac{KL} divergence}
    
    Using Lemma \ref{lemma1}, we majorize the cost function in equation (\ref{eq:CF_KL}) by the following surrogate function
    \begin{align}
    \label{eq:surrogate_fct_KL}
        g(\mathbf{\Bar{w}}_{\theta} | \mathbf{\Bar{w}}_{\theta}^{(t)}) 
    &= - \text{Re}\left( \mathbf{\Bar{w}}_{\theta}^H  2 \left[ \left( \left( - \mathbf{A} \right)\circ  \mathbf{\Sigma}_{pn} \right) \mathbf{w}_{\theta} \right. \right. \notag \\
    &\left. \left. \qquad - \left[ \mathbf{M} - \lambda_{\max}^{\mathbf{M}} \mathbf{I}_k \right] \mathbf{\Bar{w}}_{\theta}^{(t)} \right] \right)
    \end{align}
    where $\lambda^{\mathbf{M}}_{max}$ corresponds to the largest eigenvalue of $\mathbf{M}$ and $\mathbf{\Sigma}_{p}$, $\mathbf{\Sigma}_{pn}$ and $\mathbf{\Sigma}_{n}$ are presented in Table  \ref{tab:tab_plug_in}.
    \label{prop:prop_surrogate_fct_KL}
    \end{proposition}
    
    \begin{proof}
    in Appendix \ref{app_MM_KL}
    \end{proof}
    
    
    \noindent
    The equivalent optimization to the problem (\ref{eq:prob_optim_fitting_online}) is 
    \begin{equation}
        \begin{array}{c l}
        \underset{\mathbf{\Bar{w}}_{\theta} \in \mathbb{T}_k}{\rm minimize}
        & 
        g(\mathbf{\Bar{w}}_{\theta} | \mathbf{\Bar{w}}_{\theta}^{(t)})
        \end{array}
        \label{eq:optim_min_MM}
    \end{equation}
    and is solved using Lemma \ref{lemma2}. The \ac{MM} algorithm is presented in the box Algorithm \ref{algo:MM_KL}.
    
    \begin{algorithm}[H]
    \caption{\ac{MM} for \ac{KL} divergence}
    \begin{algorithmic}[1]
        \State \textbf{Input}: $\mathbf{\widetilde{\Sigma}} \in \mathbb{C}^l$, $\mathbf{w}_{\theta}^{(0)} \in \mathbb{T}_l$
        \State Compute : $\mathbf{M} = (\mathbf{D}^{-1} \circ \mathbf{\Sigma}_{n})$ and $\lambda_{max}^{\mathbf{M}}$
        \State Compute $ \mathbf{N} = 2  \left[ \left( - \mathbf{A} \right) \circ \mathbf{\Sigma}_{pn}\right] \mathbf{w}_{\theta} $
        \Repeat
            \State Compute $\ddot{\mathbf{w}}_{\theta}^{(t)} = \mathbf{N} - (\mathbf{M} - \lambda_{max}^{\mathbf{M}} \mathbf{I}) \bar{\mathbf{w}}_{\theta}^{(t)} $
            \State Update of $\mathbf{\bar{w}}_{\theta}^{(t)} = \Phi_{\mathbb{T}}\{ \ddot{\mathbf{w}}_{\theta}^{(t)} \} $ 
            \State $t = t + 1$
        \Until{convergence}
        \State \textbf{Output}: $\hat{\mathbf{w}}_{\theta} = \mathbf{w}_{end} \in \mathbb{T}_k$
    \end{algorithmic}
    \label{algo:MM_KL}
    \end{algorithm}

\subsection{Frobenius norm}
\label{subsection:Frob_seq}
Similar to the \ac{KL} divergence, this section is also divided into two parts: the first part addresses the cost functions, followed by the second part that provide the optimization algorithm.
We provide a generic form of the cost function for the Frobenius norm that can be applied to any covariance matrix plug-in.

\begin{proposition}
\textbf{Generic cost function for the Frobenius norm}

\noindent
The cost function for the Frobenius norm is represented in the following form
    \begin{align}
        f^{\text{LS}}_{\mathbf{\widetilde{\Sigma}}}(\mathbf{\widetilde{w}}_{\theta})  
        &= - 2 \mathbf{w}_{\theta}^H (|\mathbf{\Sigma}_{p}| \circ \mathbf{\Sigma}_{p}) \mathbf{w}_{\theta} \notag \\
        &\qquad - 2 \mathbf{w}_{\theta}^H (|\mathbf{\Sigma}_{pn}|)^T \circ (\mathbf{\Sigma}_{pn})^H ) \mathbf{\Bar{w}}_{\theta} \notag \\
        &\qquad - 2 \mathbf{\Bar{w}}_{\theta}^H (|\mathbf{\Sigma}_{pn}| \circ \mathbf{\Sigma}_{pn} ) \mathbf{w}_{\theta} \notag \\
        &\qquad - 2 \mathbf{\Bar{w}}_{\theta}^H (|\mathbf{\Sigma}_{n}| \circ \mathbf{\Sigma}_{n}) \mathbf{\Bar{w}}_{\theta}
        \label{eq:CF_Frob}
    \end{align}
\noindent
where $\mathbf{\Sigma}_{p}$, $\mathbf{\Sigma}_{pn}$ and $\mathbf{\Sigma}_{n}$ are presented in Table \ref{tab:tab_plug_in}.
    \label{prop:prop_CF_Frob}
\end{proposition}
\begin{proof}
in Appendix \ref{app_CF_LS}
\end{proof}

\noindent
As for the \ac{KL} divergence, the Frobenius norm has a quadratic form that enables the use of the \ac{MM} algorithm.
Using Lemma \ref{lemma3}, the cost function in equation (\ref{eq:CF_Frob}) can be majored by a function $g(\mathbf{\Bar{w}}_{\theta} | \mathbf{\Bar{w}}_{\theta}^{(t)})$ which will be provided in what follows, where we present a generic form that can be applied to any covariance matrix plug-in.

\begin{proposition}
\textbf{Surrogate function for Frobenius norm}

\noindent
The cost function can be majored on $\mathbb{T}_k$ by the following surrogate function
\begin{align}
    g(\mathbf{\Bar{w}}_{\theta} \mid \mathbf{\Bar{w}}_{\theta}^{(t)}) &= - \text{Re}\bigg(\mathbf{\Bar{w}}_{\theta}^H \cdot 4 \big[ (|\mathbf{\Sigma}_{pn}| \circ \mathbf{\Sigma}_{pn}) \mathbf{w}_{\theta} \notag \\
    &\qquad\qquad + (|\mathbf{\Sigma}_{n}| \circ \mathbf{\Sigma}_{n}) \mathbf{\Bar{w}}_{\theta}^{(t)} \big] \bigg)
    \label{eq:surrogate_fct_Frob}
\end{align}

where  $\mathbf{\Sigma}_{pn}$ and $\mathbf{\Sigma}_{n}$ are presented in Table \ref{tab:tab_plug_in}.
\label{prop:prop_surrogate_fct_Frob}
\end{proposition}
\begin{proof}
    in Appendix \ref{app_MM_LS}
\end{proof}
\noindent
Then the minimization problem can be solved using Lemma \ref{lemma2}. The \ac{MM} algorithm is presented in the box Algorithm \ref{algo:MM_Frob}.

\begin{figure*}[!ht] 
    \centering
    \begin{minipage}[b]{0.3\textwidth}
        \centering
        \includegraphics[trim={0 0 0 1.35cm}, clip,width=1.2\linewidth]{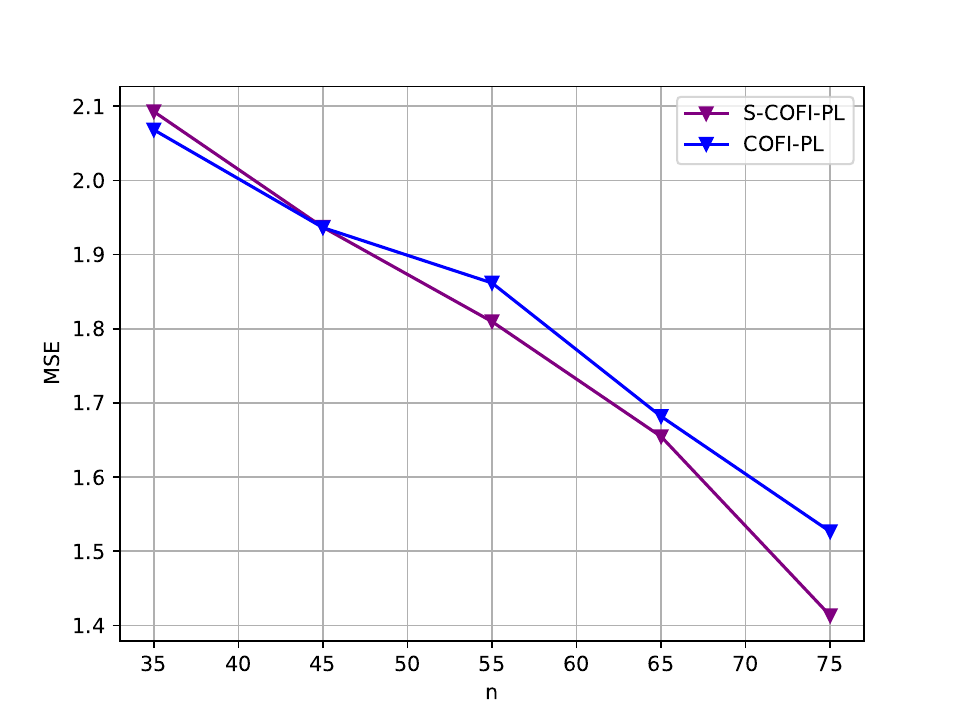}
        
        \caption*{(a)}
    \end{minipage}
    \hfill
    \begin{minipage}[b]{0.3\textwidth}
        \centering
        \includegraphics[trim={0 0 0 1.35cm}, clip, width=1.2\linewidth]{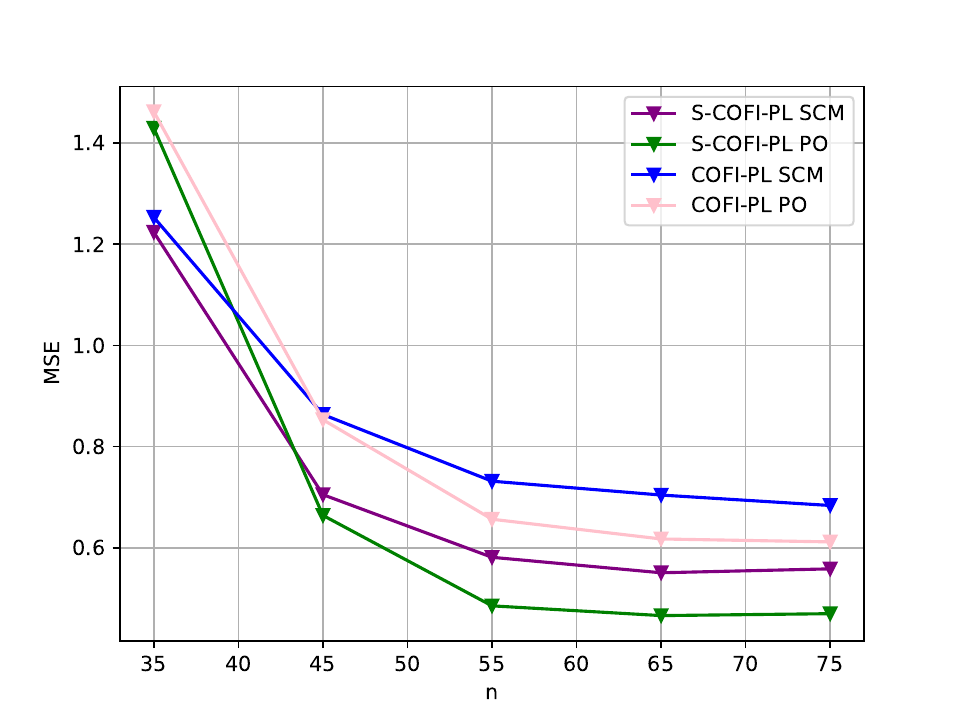}
        \caption*{(b)}
    \end{minipage}
    \hfill
    \begin{minipage}[b]{0.3\textwidth}
        \centering
        \includegraphics[trim={0 0 0 1.35cm}, clip,width=1.2\linewidth]{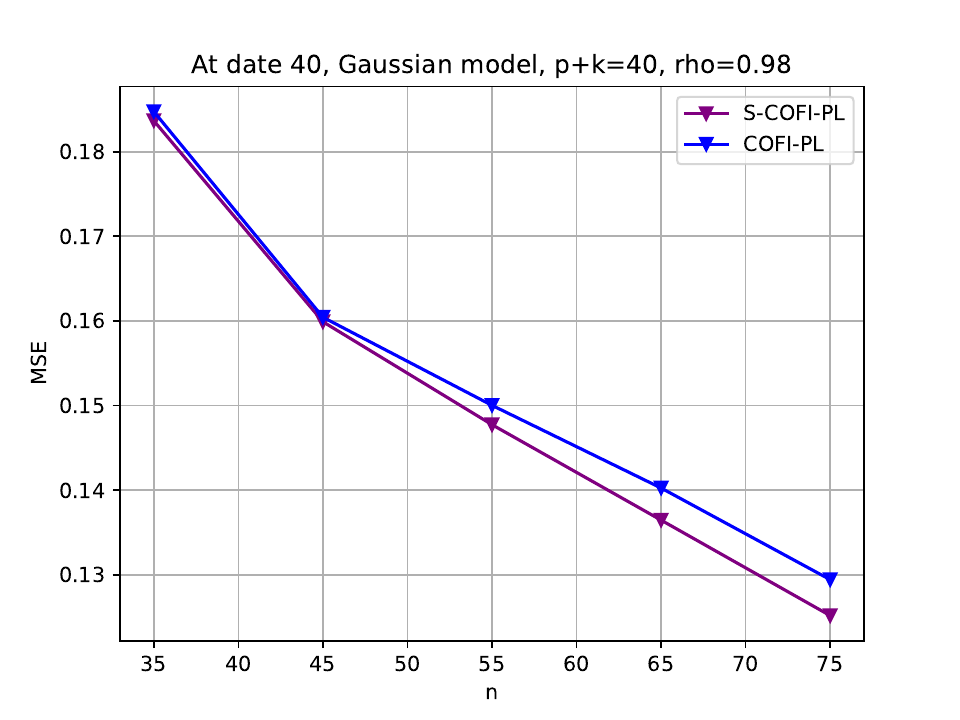}
        \caption*{(c)}
    \end{minipage}
    \caption{\ac{MSE} of phase difference estimation between the $1^{\text{st}}$ and $40^{\text{th}}$ (last) dates for $l = 40$  with respect to increasing sample size ($n$) with \ac{KL} divergence. $1^{st}$ column: \ac{SCM} using gaussian simulated data, $2^{nd}$ column: \ac{SCM} and \ac{PO} using non-gaussian simulated data, $3^{rd}$ column: \ac{SK-PO} to identity with $\beta = 0.9$ using gaussian simulated data.} 
    \label{fig:MSE_KL}
\end{figure*}

\begin{algorithm}[H]
\caption{\ac{MM} for Frobenius norm }
\begin{algorithmic}[1]
    \State \textbf{Input}: $\mathbf{\widetilde{\Sigma}} \in \mathbb{C}^l$, $\mathbf{w}_{\theta}^{(1)} \in \mathbb{T}_l$
    \State Compute  $\mathbf{M} = |\mathbf{\Sigma}_{pn}| \circ \mathbf{\Sigma}_{pn}$ 
    \State Compute  $\mathbf{N} = |\mathbf{\Sigma}_{n}| \circ \mathbf{\Sigma}_{n}$ 
    \Repeat
        \State Compute $\ddot{\mathbf{w}}_{\theta}^{(t)} = \mathbf{M} + \mathbf{N} \mathbf{\bar{w}}_{\theta}^{(t)} $
        \State Update of $ \mathbf{\bar{w}}_{\theta}^{(t)} = \Phi_{\mathbb{T}}\{ \ddot{\mathbf{w}}_{\theta}^{(t)} \} $ 
        \State $t = t + 1$
    \Until{convergence}
    \State \textbf{Output}: $\hat{\mathbf{w}}_{\theta} = \mathbf{\bar{w}}^{(\text{end})} \in \mathbb{T}_k$
\end{algorithmic}
\label{algo:MM_Frob}
\end{algorithm}

\subsection{Complexity comparison study}
\label{subsec:complexity}
\ac{PL} methods are considered to be costly, especially as they handle matrices whose size depends on the length of the time series of images used. 
For the \ac{KL} divergence, the computational complexity depends on matrix inversion and \ac{SVD} decomposition. In \ac{S-COFI-PL}, these operations are performed on $(k, k)$ matrices, while in \ac{COFI-PL}, they are performed on $(p+k, p+k)$ matrices. 
For the Frobenius norm, where matrix inversion is no longer needed, the complexity is based on matrix multiplications. In \ac{S-COFI-PL}, this involves matrices of size $(k, k)$ and $(k, p)$, whereas in the \ac{COFI-PL} method, it involves matrices of size $(p+k, p+k)$. The complexity of the presented methods are reported in Table \ref{tab:tab_complexity}. 

\begin{table}[!h]
\centering
\begin{tabular}{|cc|c|}
\hline
\multicolumn{2}{|c|}{Method} & Complexity \\ \hline
\multicolumn{2}{|c|}{S-MLE-PL \cite{elhajjar2024}} & $O((l-1)^3)$ \\ \hline
\multicolumn{2}{|c|}{Sequential Estimator \cite{ansari2017sequential}} & 
\parbox[c]{3cm}{\centering $O(5 (k^3 + m^3)$ \\ $+ (k^2 + m^2))$} \\ \hline
\multicolumn{1}{|c|}{\multirow{2}{*}{\ac{S-COFI-PL}}} & \ac{KL} divergence & $O(2 k^3 + k^2 + kp)$ \\ \cline{2-3} 
\multicolumn{1}{|c|}{} & Frobenius norm &  $O(k p + k^2)$  \\ \hline
\multicolumn{1}{|c|}{\multirow{2}{*}{\ac{COFI-PL} \cite{vu2023covariance}}} & \ac{KL} divergence  & $O(2 l^3 + l^2)$  \\ \cline{2-3} 
\multicolumn{1}{|c|}{} &  Frobenius norm & $O(l^2)$  \\ \hline
\end{tabular}
\caption{Complexity comparison for S-MLE-PL \cite{elhajjar2024}, Sequential Estimator \cite{ansari2017sequential} where $m$ represents the number of compressed images, \ac{S-COFI-PL} and \ac{COFI-PL} for both \ac{KL} divergence and the Frobenius norm.}
\label{tab:tab_complexity}
\end{table}

%% file: Sections/experiments.tex
\section{Numerical Experiments}
\label{section:exp} 

\begin{figure*}[!ht] 
    \centering
    \begin{minipage}[b]{0.3\textwidth}
        \centering
        \includegraphics[width=1.2\linewidth]{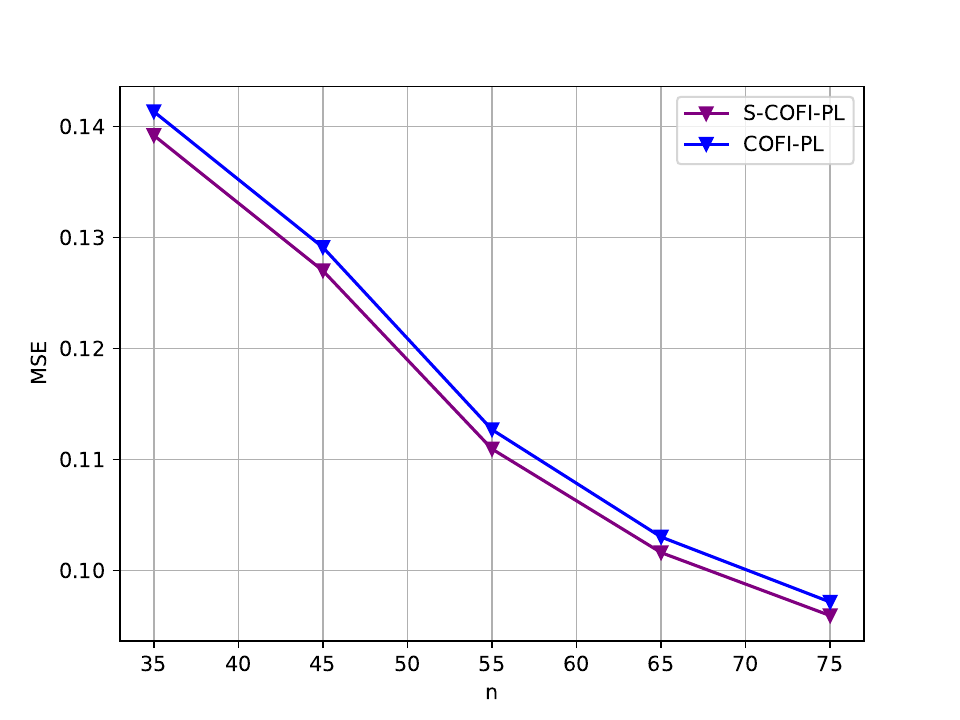}
    \end{minipage}
    \hfill
    \begin{minipage}[b]{0.3\textwidth}
        \centering
        \includegraphics[width=1.2\linewidth]{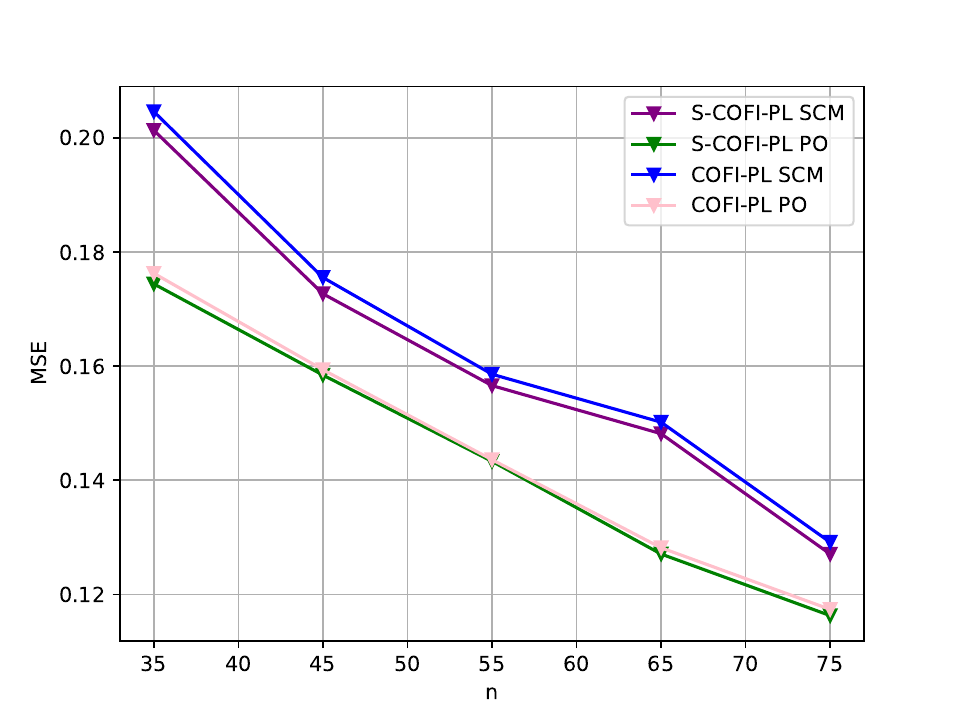}
    \end{minipage}
    \hfill
    \begin{minipage}[b]{0.3\textwidth}
        \centering
        \includegraphics[width=1.2\linewidth]{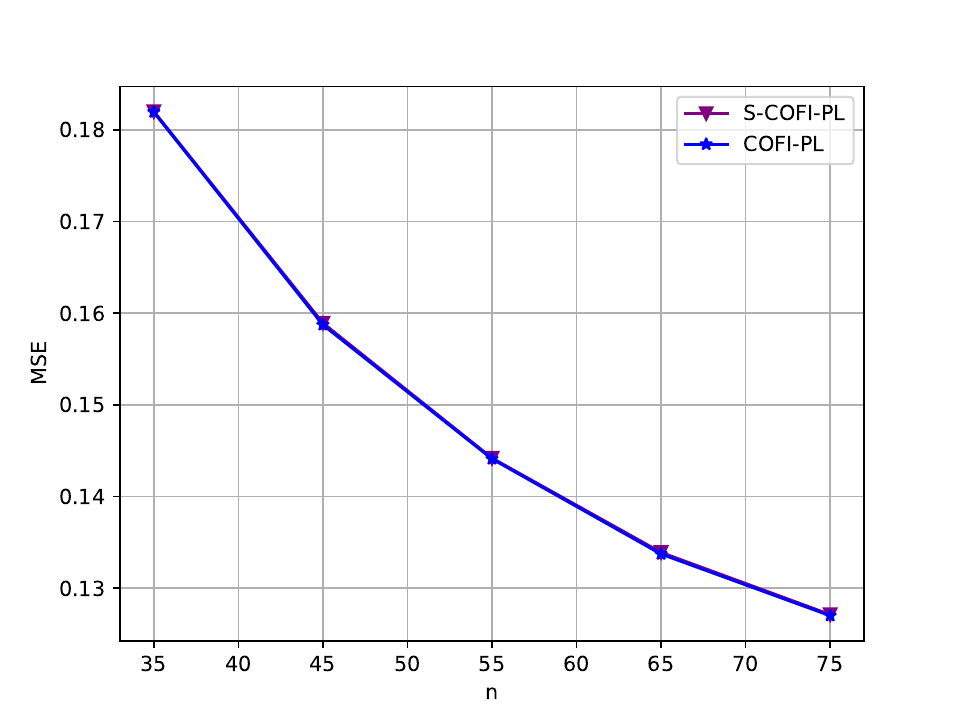}
    \end{minipage}
    \caption{\ac{MSE} of phase difference estimation between the $1^{\text{st}}$ and $40^{\text{th}}$ (last) dates for $l = 40$  with respect to increasing sample size ($n$) with Frobenius norm. $1^{st}$ column: \ac{SCM} using gaussian simulated data, $2^{nd}$ column: \ac{SCM} and \ac{PO} using non-gaussian simulated data, $3^{rd}$ column: \ac{BW-PO} with $b = 9$ using gaussian simulated data.} 
    \label{fig:MSE_LS}
\end{figure*}

We simulate a time series of size $l = p+k = 40$ images, where $p = 35$ and $k = 5$. The real core of the covariance matrix  $\mathbf{\tilde{\Psi}}$ is simulated as a Toeplitz matrix i.e.,  $[ \mathbf{\tilde{\Psi}}]_{ij} = \rho^{|i-j|}$ with a coefficient correlation $\rho = 0.98$. Phases differences vary linearly between $0$ and $2$ rad, i.e., $\Delta_{i, i-1} = \theta_i - \theta_{i-1} = 2/l$ rad.  The covariance matrix is then obtained according to equation (\ref{cov_mat_struc}). We consider two scenarios and we simulate $n$ \ac{i.i.d} samples according to
\begin{itemize}
    \item Gaussian distribution assumption: $\mathbf{\widetilde{x}}^i \sim \mathcal{N}(0,\, \mathbf{\widetilde{\Sigma}})$
    \item Non Gaussian distribution assumption: $\mathbf{\widetilde{x}}^i \sim \mathcal{N}(0,\, \tau_i \mathbf{\widetilde{\Sigma}})$ , where each $\tau_i$ is sampled following a Gamma distribution $\tau \sim \Gamma(\nu, \frac{1}{\nu})$ with $\nu = 1$.
\end{itemize}
We compare the results of our approach with the offline \ac{COFI-PL} approach \cite{vu2024covariance}. Several matrix distances and covariance matrix plug-in are used and presented. \acl{MSE} (\acs{MSE}) are computed using $1000$ Monte Carlo trials.

\noindent
Fig. \ref{fig:MSE_KL} provides the \ac{MSE} of phase estimate when the sample size $n$ increases. 
Fig. \ref{fig:MSE_KL} (a) shows a comparison of phase estimates with \ac{SCM} as a plug-in for the covariance matrix between the offline \ac{COFI-PL} approach and the sequential \ac{S-COFI-PL} approach. The phase difference estimations from these 2 approaches yield similar accuracy. 
Fig. \ref{fig:MSE_KL} (b) can be analysed in $2$ different setups, comparing the offline and sequential approaches, and second, comparing the plug-in \acs{SCM} with the \ac{PO}. For this case, we used the data simulated according to a Non Gaussian distribution. It is known that , in the context of covariance matrix estimation, robust estimators behave equivalently to estimators based on a Gaussian model assumption when the data follow a Gaussian distribution \cite{pascal2008empirical, harari2008use}.  However, when the data follow a Non Gaussian distribution, the performance improvement is significant \cite{vu2023robust}, as it is showed in Fig. \ref{fig:MSE_KL} (b).
Fig. \ref{fig:MSE_KL} (c) shows a comparison between the sequential and offline processing with the shrinkage to identity regularization for the covariance matrix plug-in. Both approaches yields similar results.

\noindent
Fig. \ref{fig:MSE_LS} represents the \ac{MSE} of phase estimates when $n$ increases using the Frobenius norm. Fig. \ref{fig:MSE_LS} (a) shows that taking the \acs{SCM} as a plug-in yields similar results for the sequential and offline approaches. As in the case of the \ac{KL} divergence, \ac{PO} yields better results than the \ac{SCM} when applied to Non Gaussian data (Fig. \ref{fig:MSE_LS} (b)).  Fig. \ref{fig:MSE_LS} (c) shows similar performances between the \ac{S-COFI-PL} and \ac{COFI-PL} when taking a tapering regularization for the \ac{PO}.


\noindent
As shown, the sequential approach \ac{S-COFI-PL} yields similar results to the offline method \ac{COFI-PL} \cite{vu2024covariance}. Nonetheless, it is interesting to study the scenario of integrating sequentially several blocks of new \ac{SAR} images, in order to analyze the error propagation in the estimated phase. Let's consider a scenario of having a total of $t$ \ac{SAR} images. We consider the following decomposition for the sequential treatment : 
\begin{itemize}
    \item $p$ \ac{SAR} images represent the first block, which will be used to apply the \ac{COFI-PL}
    \item $k$ \ac{SAR} images represent the second block, which will be used to apply the \ac{S-COFI-PL} approach along with the first block
    \item $m$ \ac{SAR} images represent the third block. The estimation of the $m$ phases will be examined in several ways. First, using the offline \ac{COFI-PL} approach where $t = p + k + m $ \ac{SAR} images are used directly. Second, the $m$ phases are estimated according to the $l = p + k$ past images which can be approached in two way: the offline method and the concatenation of $p$ phases computed offline with $k$ phases calculated sequentially.
\end{itemize}

\begin{figure}[H]
    \centering
    \includegraphics[trim={0 0 0 1.35cm}, clip,width=7cm]{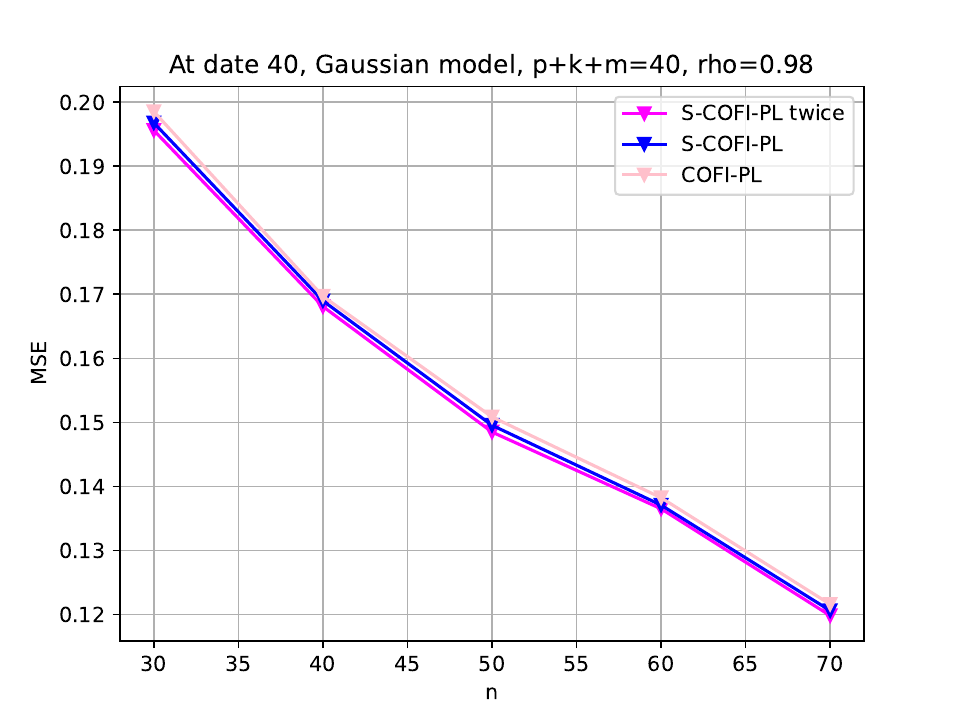}
    \caption{\small \ac{MSE} on $\mathbf{\Bar{w}}_{\theta}$ with increasing $n$, $t = 40$, $\rho=0.98$ using $1000$ Monte Carlo trials.}
    \label{fig:multiple_bloc}
\end{figure}

\noindent
Fig. \ref{fig:multiple_bloc} represents \ac{MSE} variations on $\mathbf{\Bar{w}}_{\theta}$ with an increase in sample size $n$. For the simulations, we consider a scenario of a total of $t = 40$ images, divided into 3 blocks. The first block is of size $p = 30$, the second block is of size $k = 5$, and the third block is of size $m = 5$.
The coherence matrix is simulated as a Toeplitz matrix with a coefficient correlation $\rho = 0.98$. \ac{PO} is used as a plug-in for the covariance matrix with the Frobenius norm as a matrix distance. 
The application once of the sequential method, \ac{S-COFI-PL}, yields results similar to the offline method, \ac{COFI-PL}. The successive integration of multiple blocks of new \ac{SAR} images, taking into account the results of the sequential method progressively, shows similar performances, in terms of \ac{MSE}, as the offline approach \ac{COFI-PL} and the sequential approach \ac{S-COFI-PL} applied once.

%% file: Sections/real_data.tex
\section{Real world study}
\label{section:real_data}

\begin{figure*}[!ht] 
    \centering
    \begin{minipage}[b]{0.45\textwidth}
        \caption*{\footnotesize (a) \ac{S-COFI-PL} PO}
        \centering
        \includegraphics[trim={0 5cm 0 5cm}, clip, width=\columnwidth, height=0.5\textwidth]{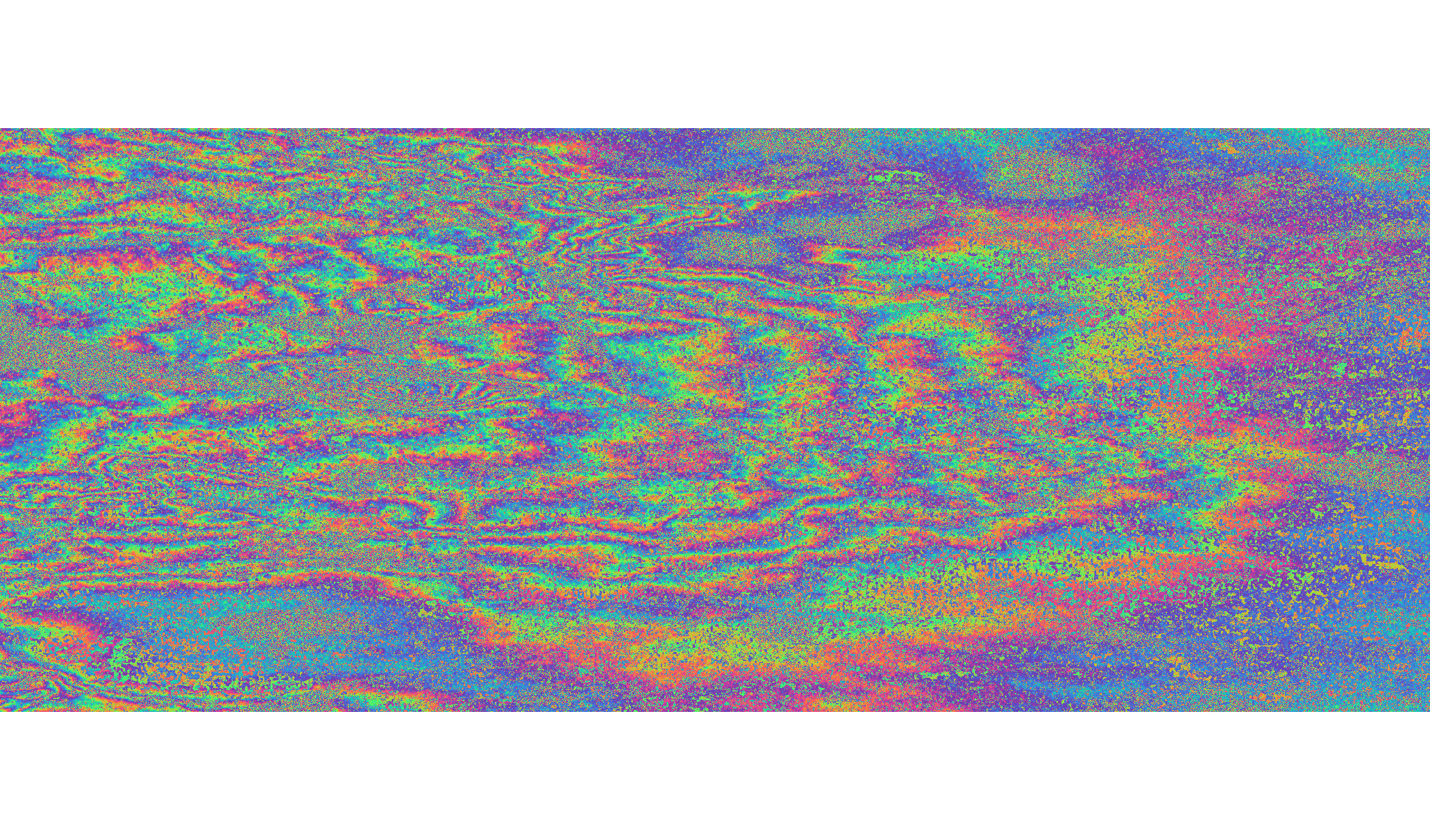} 
    \end{minipage}
    \hspace{1cm} 
    \begin{minipage}[b]{0.45\textwidth}
        \caption*{\footnotesize (b) \ac{COFI-PL} PO}
        \centering
        \includegraphics[trim={0 2cm 0 2cm}, clip, width=\columnwidth, height=0.5\textwidth]{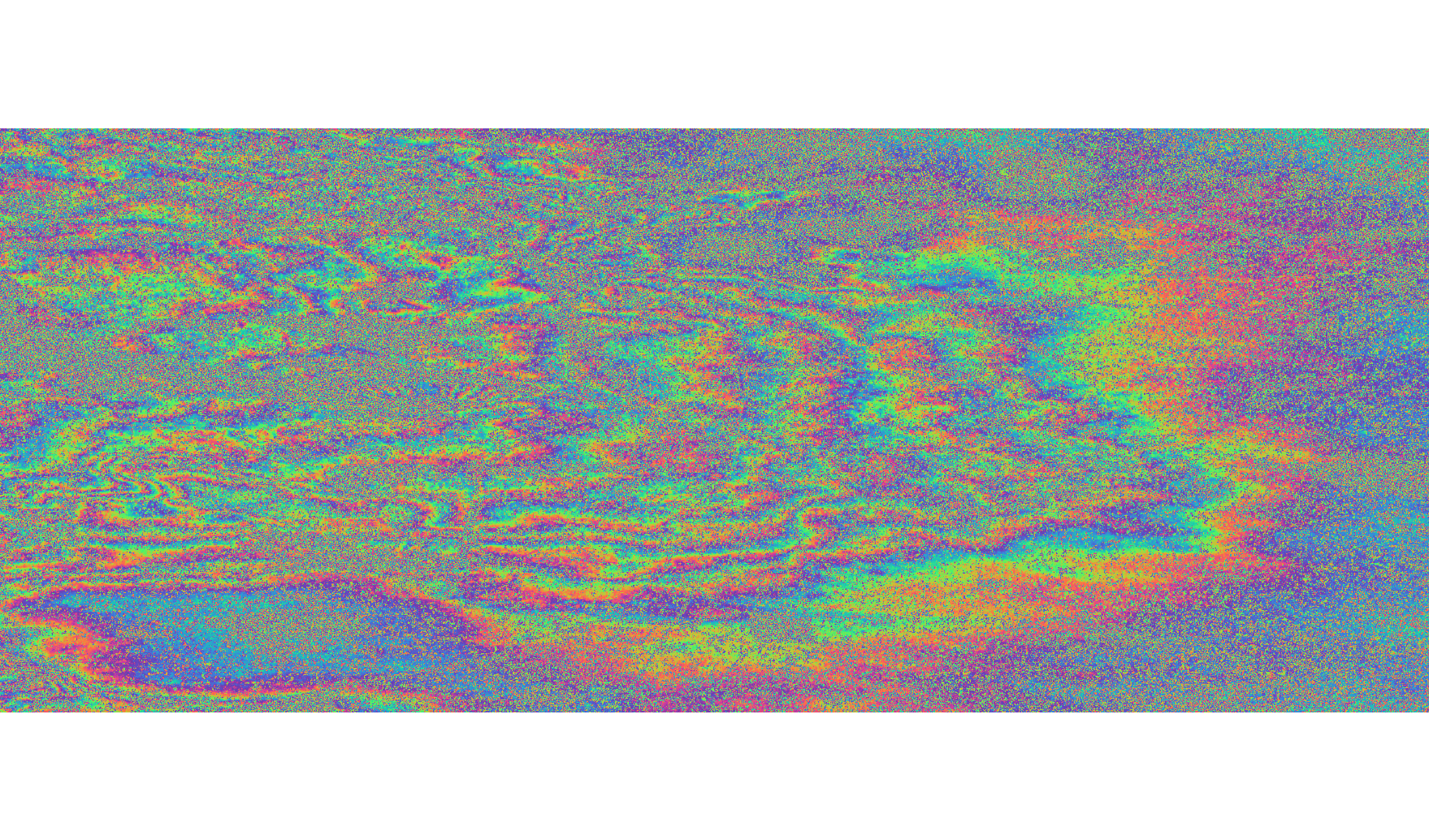} 
    \end{minipage}
    \begin{minipage}[b]{0.45\textwidth}
        \caption*{\footnotesize (e) \ac{S-COFI-PL} SK-PO}
        \centering
        \includegraphics[trim={0 5cm 0 5cm}, clip, width=\columnwidth, height=0.5\textwidth]{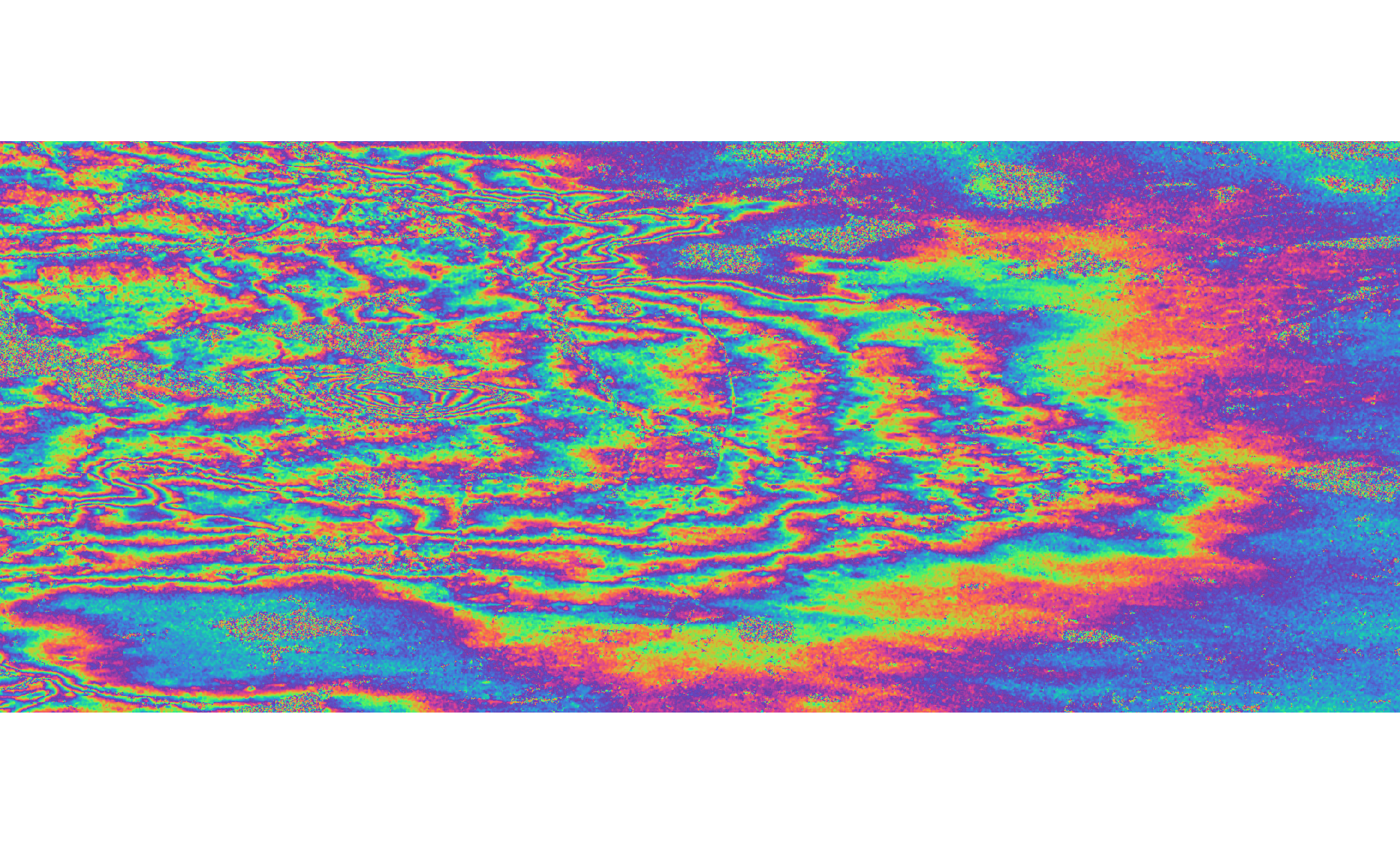} 
    \end{minipage}
    \hspace{1cm} 
    \begin{minipage}[b]{0.45\textwidth}
        \caption*{\footnotesize (f) \ac{COFI-PL} SK-PO}
        \centering
        \includegraphics[trim={0 6cm 0 6cm}, clip, width=\columnwidth, height=0.5\textwidth]{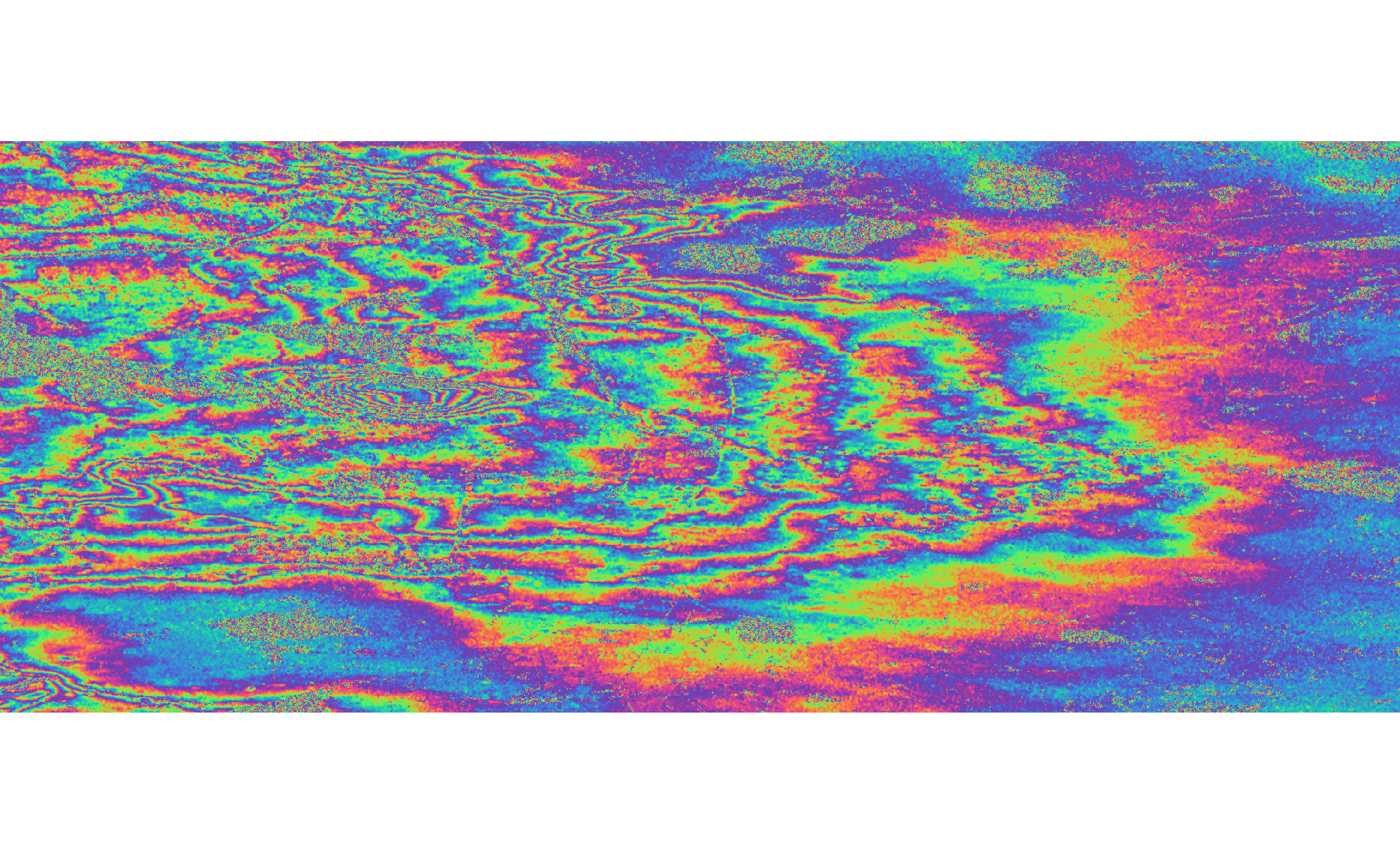} 
    \end{minipage}
    
    \hspace{1.5cm} 
    \begin{minipage}[b]{\textwidth}
        \centering
        \includegraphics[width=0.5\textwidth, height=4em]{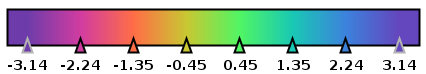}
    \end{minipage}
    
    \caption{The longest temporal baseline interferograms estimated by \ac{S-COFI-PL} ($1^{\text{st}}$ column) and \ac{COFI-PL} \cite{vu2024covariance} ($2^{\text{nd}}$ column) with various plug-in estimators for the covariance matrix and \ac{KL} divergence. $1^{\text{st}}$ row: \ac{PO}, $2^{\text{nd}}$  row: SK-PO}
    \label{fig:real_data_KL}
\end{figure*}

\subsection{Study area and dataset}


Ranking among the five largest cities in the world, Mexico city has a population exceeding $20$ million habitants. It covers $2000$ $\text{km}^2$, $2300$ m above sea level. 
Mexico city went through rapid urbanization, which led to a significant increase in demand for water. The primary water supply comes from aquifers, leading to subsidence and deformation across the city. Mexico city consists of an interesting case study for many multi-temporal \ac{InSAR} approaches. We used a data set of $40$ \ac{SAR} images over Mexico city, acquired between $14$ August $2019$ and $18$ December $2021$, every $12$ days. Pre-processing treatment was done via SNAP software \cite{esa_snap} where all images are co-registered with reference to the first date $14$ August $2019$.
In order to relate to the parameters in the simulations, $p = 35$, $k = 5$, the sample size is set to $n = 64$ (i.e sliding window of size $8 \times 8$ pixels).

\noindent
For this study, we focus on the comparison of the results of \ac{S-COFI-PL} with the offline approach \ac{COFI-PL} \cite{vu2023covariance, vu2024covariance}, and the sequential approaches proposed in \cite{ansari2017sequential, elhajjar2024}. 

\subsection{Qualitative assessment of the estimated phase}

 In Fig. \ref{fig:real_data_KL} and Fig.  \ref{fig:real_data_frob}, we show only the estimation of the interferogram with the longest temporal baseline of $857$ days, to emphasize performance in relation to temporal decorrelation. 
We provide results for various covariance matrix estimations, however our objective is to compare the outcomes of the sequential and offline approaches, regardless the chosen plug-in for the covariance matrix. Detailed comparison is made in \cite{vu2024covariance} between the different matrix distances that can be used, as well as the different covariance matrix plug-in. 

\noindent
Fig. \ref{fig:real_data_KL} presents the results for the \ac{KL} divergence. As highlighted in \cite{vu2023covariance}, the choice of the plug-in estimator significantly impacts the quality of the resulting interferogram (see the comparison in rows in Fig. \ref{fig:real_data_KL}). The \ac{PO} plug-in shows superior performance in the sequential approach \ac{S-COFI-PL} (Fig. \ref{fig:real_data_KL}  (a)) compared to the offline method \ac{COFI-PL} (Fig.  \ref{fig:real_data_KL} (b)). The \ac{KL} divergence function requires matrix inversion, which becomes increasingly complex as the matrix size grows. This issue is less problematic in the sequential processing, as the inversion is applied to smaller matrices, thereby enhancing the results compared to the offline processing. Additionally, after closer inspection, the shrinking to identity regularization (SK-PO) demonstrates superior performance in the sequential processing \ac{S-COFI-PL} compared to the offline \ac{COFI-PL} (Fig.  \ref{fig:real_data_KL}  (c) and (d)).
\noindent
The performances of \ac{S-COFI-PL} and \ac{COFI-PL}, in Fig.  \ref{fig:real_data_frob},  are similar. A comparison of the interferograms in the two columns reveals comparable quality and noise levels throughout the scene.
Additionally, a comparison of the different plug-ins is illustrated in Fig.  \ref{fig:real_data_frob} (a), (b), with the \ac{PO} plug-in, Fig. \ref{fig:real_data_frob} (c), (d), with the shrinking of the \ac{PO} plug-in, and  Fig. \ref{fig:real_data_frob} (e), (f) with the tapering regularization of the \ac{PO} plug-in. Upon closer examination, as mentioned in \cite{vu2023covariance}, regularization, particularly tapering, yields superior results.

\begin{figure*}[!ht] 
    \centering
    \begin{minipage}[b]{0.45\textwidth}
        \caption*{\footnotesize (a) \ac{S-COFI-PL} PO}
        \centering
        \includegraphics[trim={0 5cm 0 5cm}, clip, width=\columnwidth, height=0.5\textwidth]{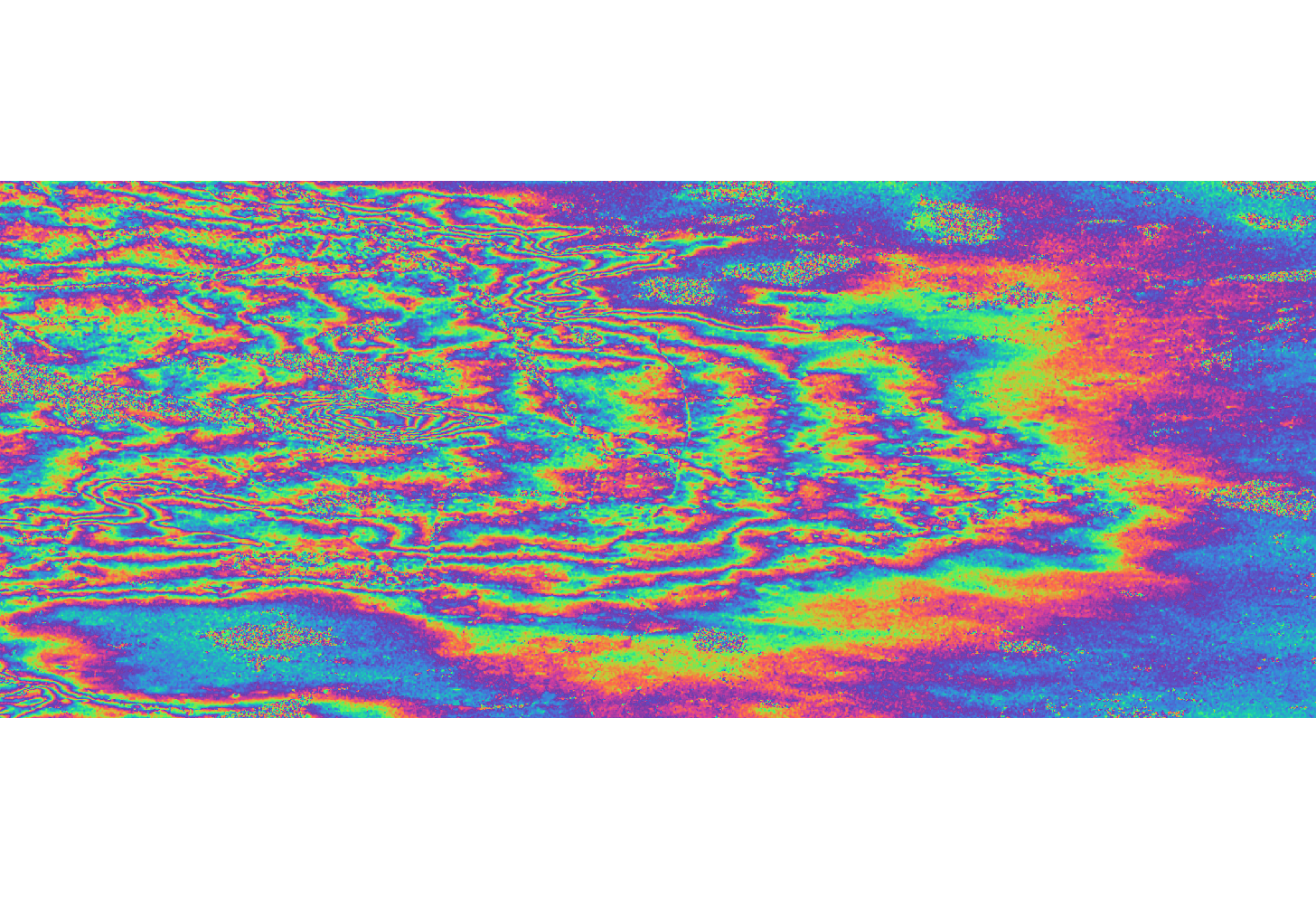} 
    \end{minipage}
    \hspace{1cm} 
    \begin{minipage}[b]{0.45\textwidth}
        \caption*{\footnotesize (b) \ac{COFI-PL} PO}
        \centering
        \includegraphics[trim={0 5cm 0 5cm}, clip, width=\columnwidth, height=0.5\textwidth]{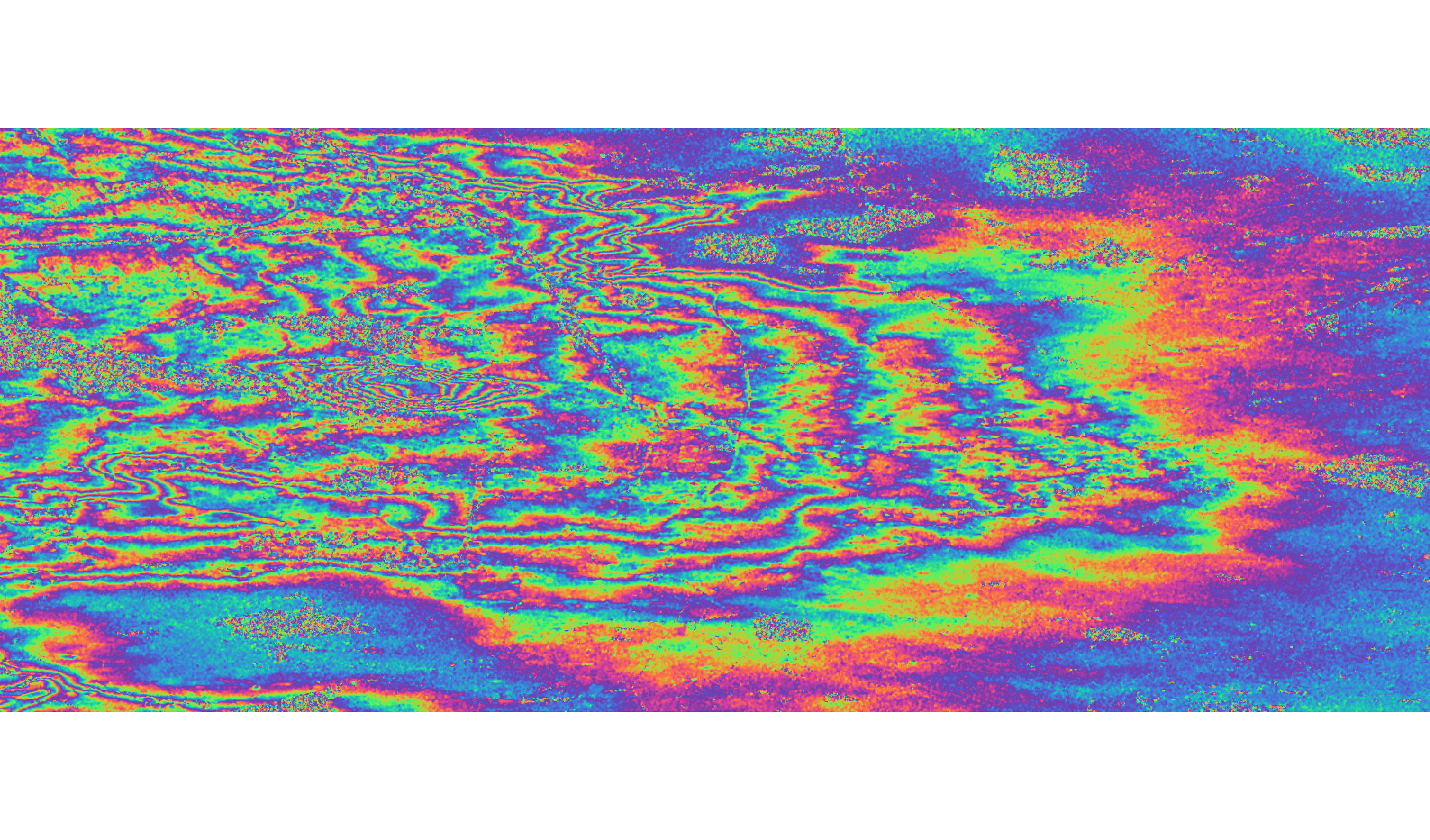} 
    \end{minipage}
    \begin{minipage}[b]{0.45\textwidth}
        \caption*{\footnotesize (c) \ac{S-COFI-PL} SK-PO}
        \centering
        \includegraphics[trim={0 2cm 0 2cm}, clip, width=\columnwidth, height=0.5\textwidth]{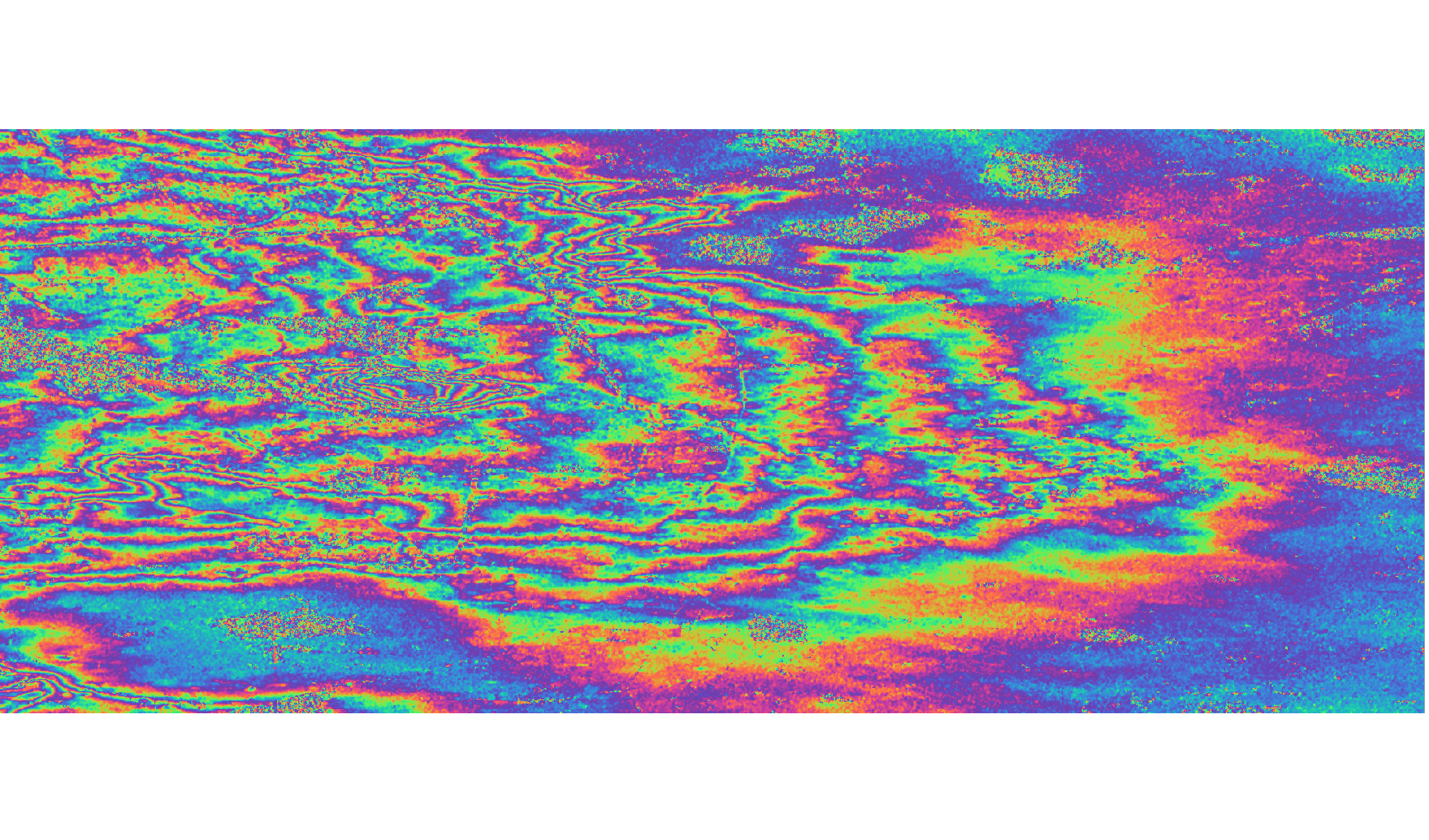} 
    \end{minipage}
    \hspace{1cm} 
    \begin{minipage}[b]{0.45\textwidth}
        \caption*{\footnotesize (d) \ac{COFI-PL} SK-PO}
        \centering
        \includegraphics[trim={0 5cm 0 5cm}, clip, width=\columnwidth, height=0.5\textwidth]{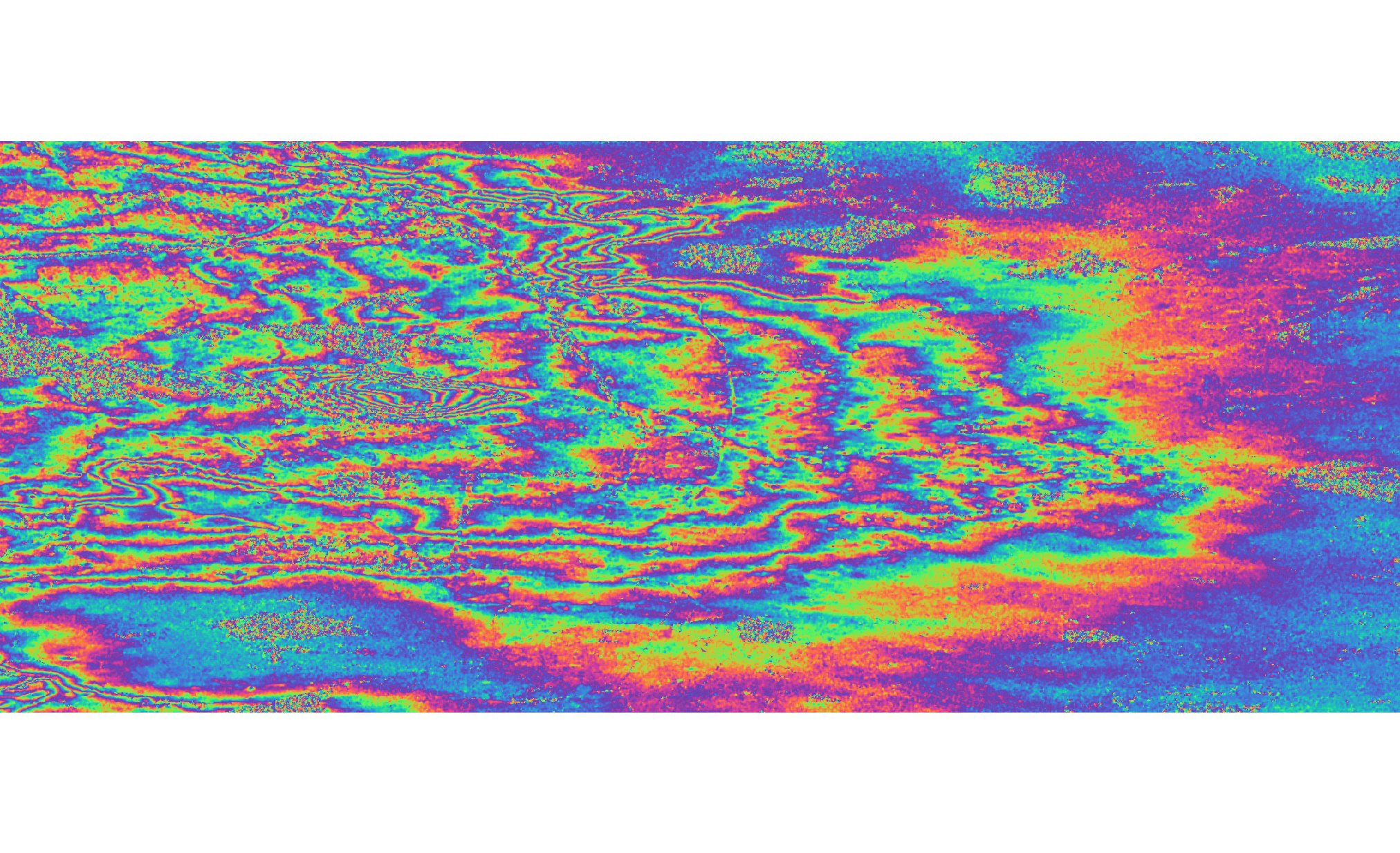} 
    \end{minipage}
    \begin{minipage}[b]{0.45\textwidth}
        \caption*{\footnotesize (e) \ac{S-COFI-PL} BW-PO}
        \centering
        \includegraphics[trim={0 5cm 0 5cm}, clip, width=\columnwidth, height=0.5\textwidth]{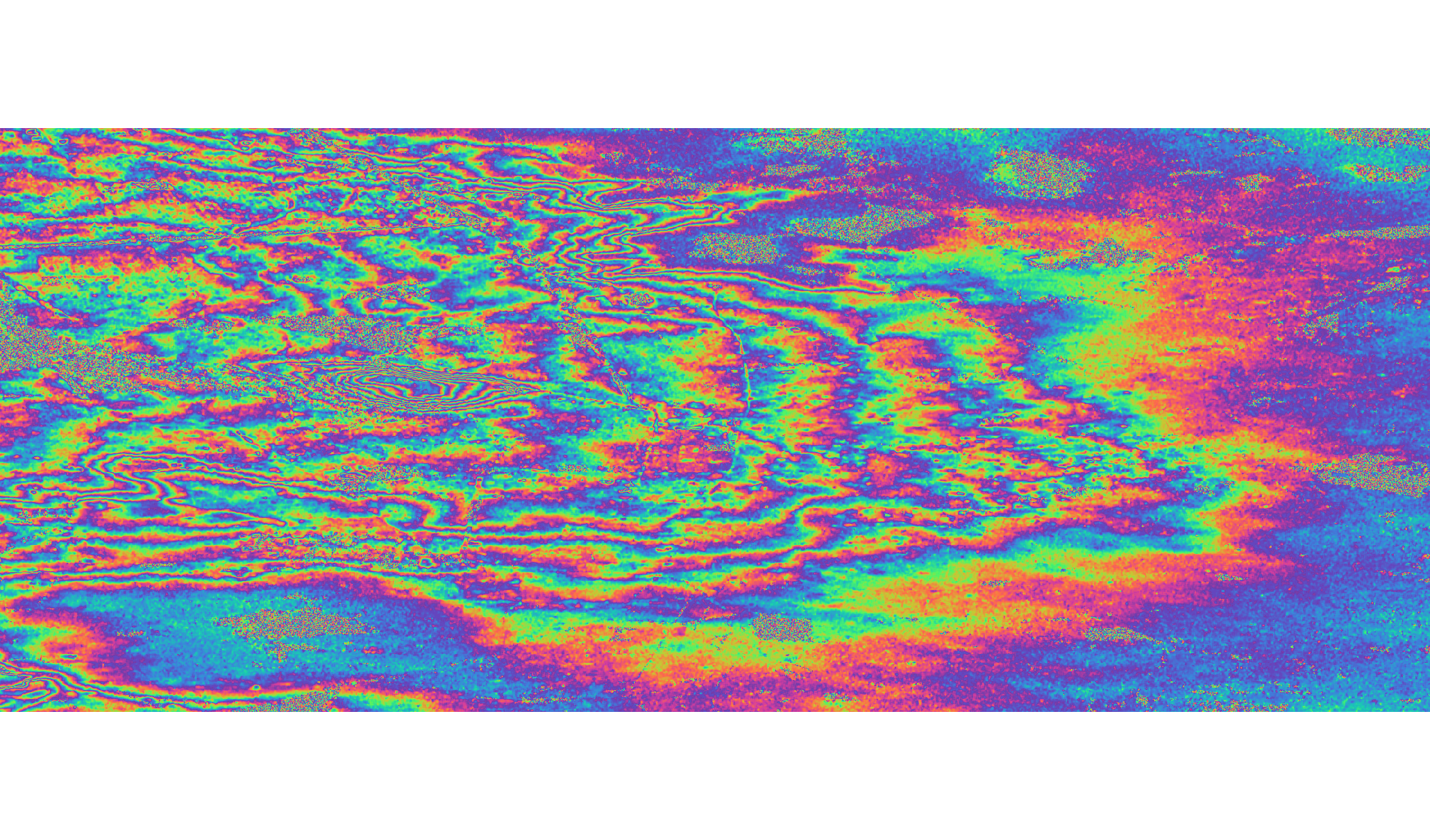} 
    \end{minipage}
    \hspace{1cm} 
    \begin{minipage}[b]{0.45\textwidth}
        \caption*{\footnotesize (f) \ac{COFI-PL} BW-PO}
        \centering
        \includegraphics[trim={0 5cm 0 5cm}, clip, width=\columnwidth, height=0.5\textwidth]{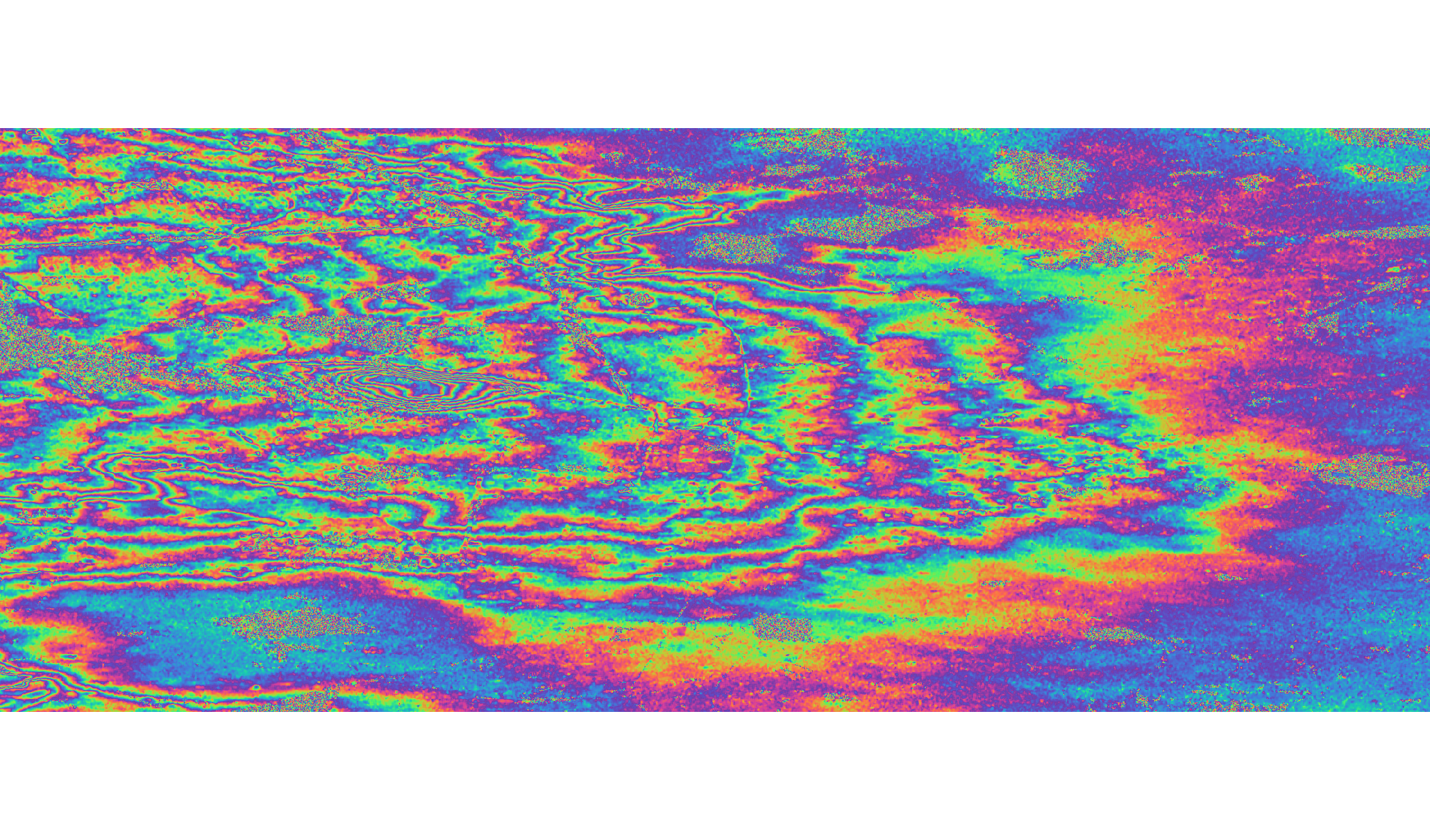} 
    \end{minipage}
    \hspace{1.5cm} 
    \begin{minipage}[b]{\textwidth}
        \centering
        \includegraphics[width=0.5\textwidth, height=4em]{images/real_data/legend_phase.png}
    \end{minipage}
    
    \caption{The longest temporal baseline interferograms estimated by \ac{S-COFI-PL} ($1^{\text{st}}$ column) and \ac{COFI-PL} \cite{vu2024covariance} ($2^{\text{nd}}$ column) with various plug-in estimators for the covariance matrix with Frobenius norm. $1^{\text{st}}$ row: \ac{SCM}, $2^{\text{nd}}$ row: SK-\ac{PO}, $3^{\text{rd}}$ row: BW-\ac{PO}}
    \label{fig:real_data_frob}
\end{figure*}

\begin{figure}[!h]
    \centering
    \begin{minipage}[b]{\columnwidth}
        \includegraphics[trim={0 2cm 0 2cm}, clip,
        width=\columnwidth]{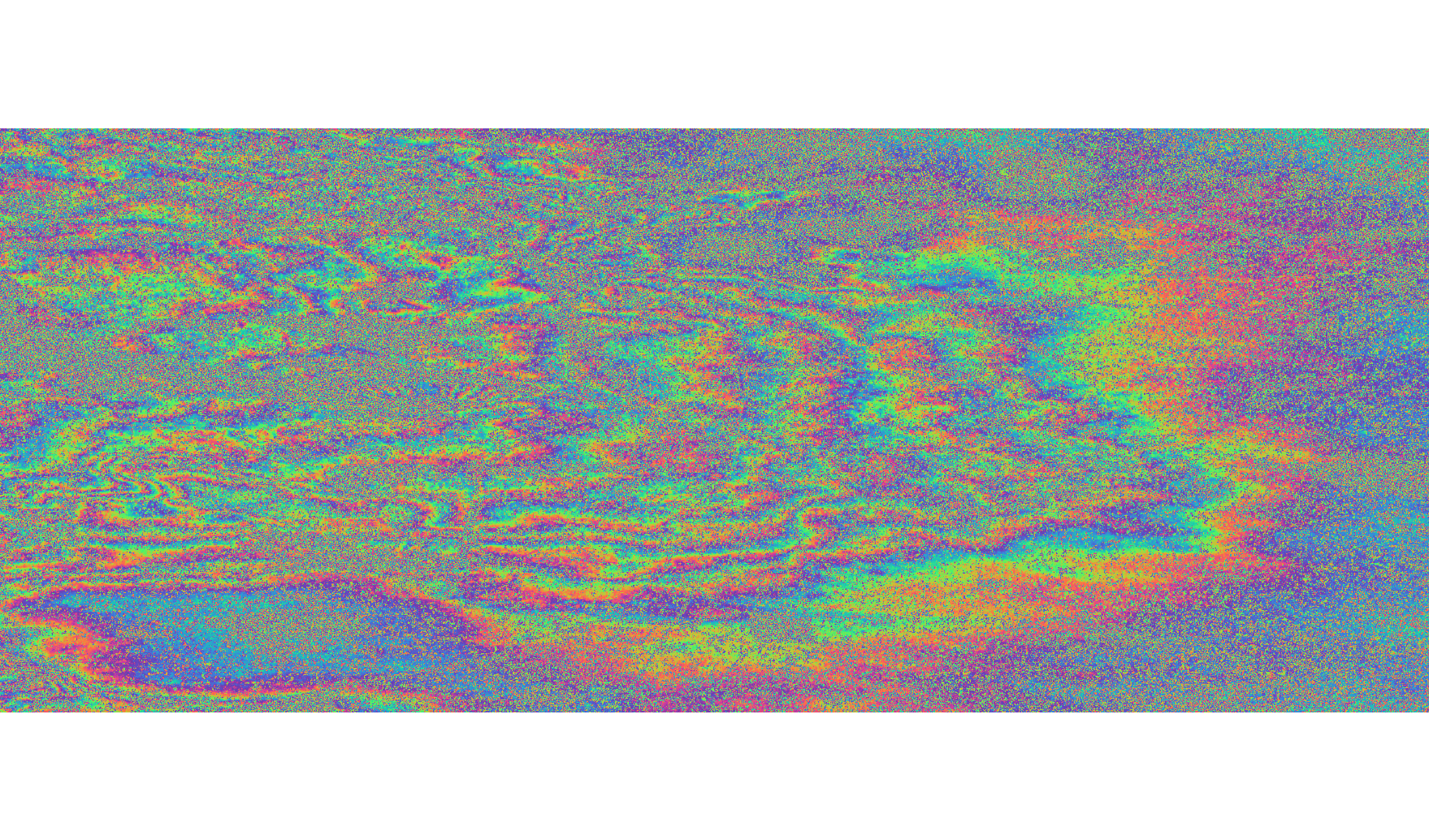}
        \caption*{(a) S-MLE-PL}
    \end{minipage}
    \vspace{2mm} 
    \begin{minipage}[b]{\columnwidth}
        \includegraphics[trim={0 5cm 0 5cm}, clip, width=\columnwidth]{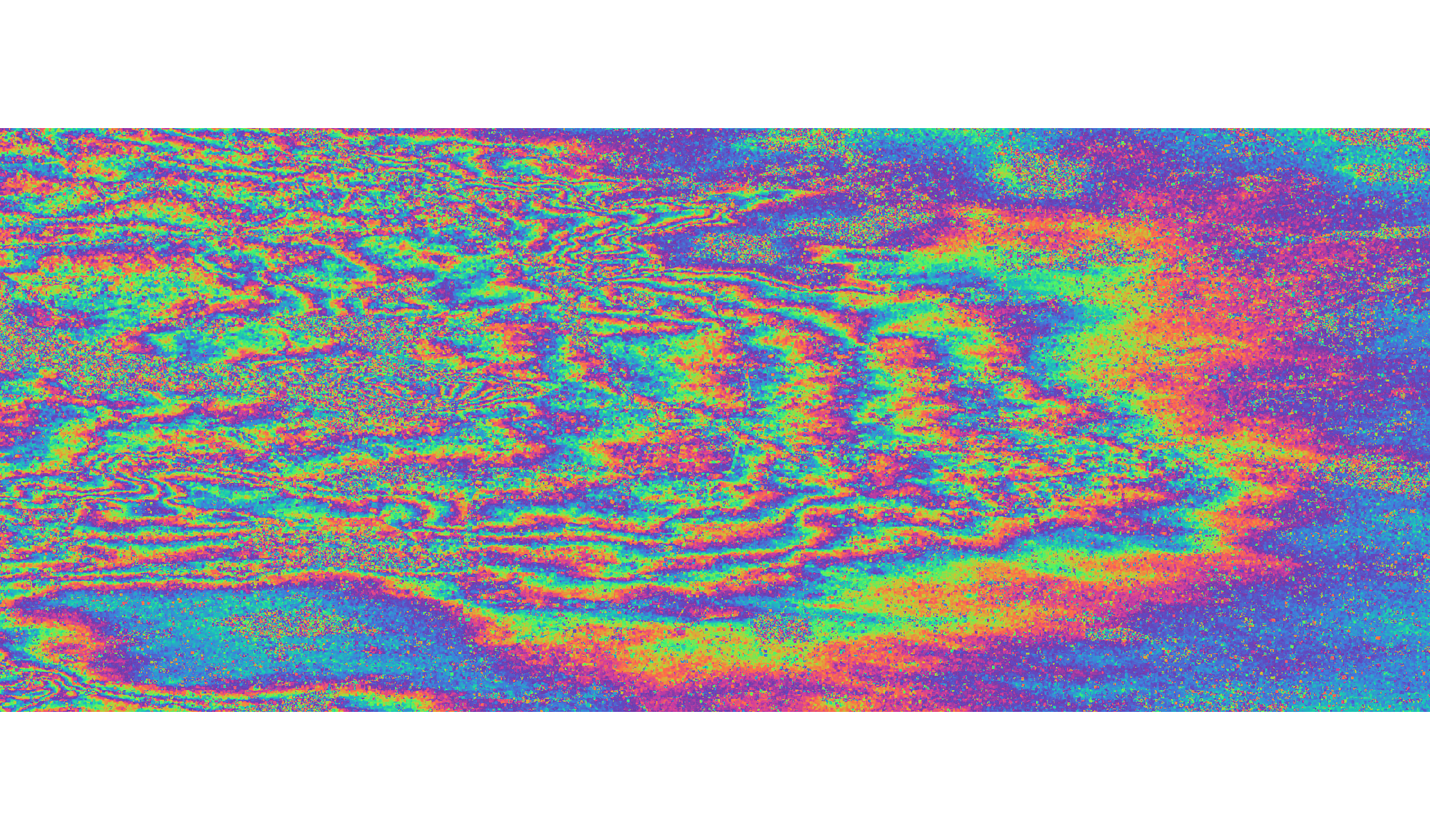} 
        \caption*{(b) Sequential Estimator}
    \end{minipage}
    \includegraphics[width=0.5\textwidth, height=4em]{images/real_data/legend_phase.png}
    \caption{The longest temporal baseline interferogram estimated by (a) S-MLE-PL \cite{elhajjar2024} and (b) Sequential Estimator \cite{ansari2017sequential} with a stack of size $k = 5$.}
    \label{fig:seq_realdata}
\end{figure}
\noindent
Fig. \ref{fig:seq_realdata} representes the results of the sequential approaches in the state-of-the-art \cite{ansari2017sequential, elhajjar2024}. 
The Sequential Estimator \cite{ansari2017sequential} shows performance comparable to that of shrinking the \ac{PO} to the identity using the \ac{KL} divergence, though the latter remains slightly noisier. This result is not surprising, as the sequential estimator can be viewed as a low-rank approximation. In \cite{vu2024covariance}, it was shown that the low-rank structure of the covariance matrix yields results similar to those obtained with shrinkage regularization to the identity. However, the sequential estimator appears to be outperformed by the \ac{PO}, SK-\ac{PO}, and BW-\ac{PO} estimators using the Frobenius norm, which are less noisy compared to the sequential estimator. The \ac{SMLEPL} proposed in \cite{elhajjar2024} is expected to yield similar results to \ac{S-COFI-PL} with \ac{SCM} as a plug-in with the \ac{KL} divergence. 

\subsection{Quantitative assessment of the estimated phase}

In Table \ref{tab:quantitative}, we show values of several quantitative criteria such as  \ac{RMSE}, \ac{SCC}, \ac{UQI}, \ac{SSIM} and the colinearity criterion \cite{pinel2012multi}.  The colinearity criterion evaluates whether the results are noisy and allows spatial comparison of the phase with its surroundings. A value closer to $1$ indicates better results.
We select the results of the offline approach \ac{COFI-PL} with the Frobenius norm and the \ac{BW-PO} regularization as our baseline, which outperforms state-of-art approaches as demonstrated in \cite{vu2024covariance}. 
The first axis of comparison is between the offline and sequential processing. For the \ac{KL} divergence, comparable results are observed between the two processing approaches, where in most cases, \ac{S-COFI-PL} outperforms \ac{COFI-PL}. As shown in Fig. \ref{fig:real_data_KL} ((a) and (b)), the sequential processing with \ac{PO} yields better results than the offline processing.

For the Frobenius norm, similar values are observed across all indicators for both sequential and offline processing with each plug-in. The colinearity of outputs is also similar in both sequential and offline approaches, consequently achieving values greater than $0.85$.
Compared to state-of-the-art sequential methods, \ac{S-COFI-PL} with the Frobenius norm delivers superior performance irrespective of the chosen plug-in. For the \ac{KL} divergence, \ac{S-COFI-PL} achieves better results with the shrinking-to-identity regularization. For other metrics, the sequential estimator \cite{ansari2017sequential} generally performs better.
With the sequential \ac{SMLEPL} method \cite{elhajjar2024}, \ac{S-COFI-PL} using the Frobenius norm yields better results across all plug-ins. For \ac{S-COFI-PL} with the \ac{KL} divergence, similar performance to \ac{BW-SCM} is observed, which is better than \ac{SCM} and \ac{BW-PO}. Otherwise, with other plug-ins and regularizations, \ac{S-COFI-PL} demonstrates superior results.

\begin{table*}[]
\centering
\begin{tabular}{|cc||c|c|c|c|c|}
\hline
\multicolumn{2}{|l|}{}                    & \ac{RMSE} & \ac{SCC} & \ac{UQI} & \ac{SSIM} & colinearity \cite{pinel2012multi} \\ \hline
\multicolumn{1}{|l||}{\multirow{6}{*}{\ac{KL} divergence}} &  \ac{S-COFI-PL} \ac{SCM}  & 2.63 & 1.00 & 0.05 & 0.0003 & 0.21\\ \cline{2-7}
\multicolumn{1}{|l||}{}                  & \ac{COFI-PL} \ac{SCM} &  2.51  & 0.93 & 0.07 & 0.001 & 0.24 \\ \cline{2-7}
\multicolumn{1}{|l||}{}                  & \ac{S-COFI-PL} \ac{PO} &  2.01  & 0.55 & 0.36 & 0.11 & 0.55\\ \cline{2-7} 
\multicolumn{1}{|l||}{}                  & \ac{COFI-PL} \ac{PO} &  2.11  & 0.63 & 0.27 & 0.06 & 0.42 \\ \cline{2-7}
\multicolumn{1}{|l||}{}                  & \ac{S-COFI-PL} \ac{SK-SCM} & 2.03  & 0.55 & 0.39 & 0.13 & 0.61\\ \cline{2-7}
\multicolumn{1}{|l||}{}                  & \ac{COFI-PL} \ac{SK-SCM} & 2.03 & 0.56 & 0.38 & 0.12 & 0.60 \\ \cline{2-7}
\multicolumn{1}{|l||}{}                  & \ac{S-COFI-PL} \ac{SK-PO} &  1.35 & 0.23 &  0.73 & 0.58 & 0.90\\ \cline{2-7}
\multicolumn{1}{|l||}{}                  & \ac{COFI-PL} \ac{SK-PO} & 1.40 & 0.25 & 0.70 & 0.53 & 0.87 \\ \cline{2-7}
\multicolumn{1}{|l||}{}                  & \ac{S-COFI-PL} \ac{BW-SCM} &  2.63 & 0.98 & 0.08 & 0.001 & 0.31\\ \cline{2-7}
\multicolumn{1}{|l||}{}                  & \ac{COFI-PL} \ac{BW-SCM} & 2.63 & 1.01 & 0.09 & 0.002 & 0.32 \\ \cline{2-7}
\multicolumn{1}{|l||}{}                  & \ac{S-COFI-PL} \ac{BW-PO} & 2.61  & 0.96 & 0.03 & 0.001 & 0.17\\ \cline{2-7}
\multicolumn{1}{|l||}{}                  & \ac{COFI-PL} \ac{BW-PO} & 2.56 & 0.96 & 0.04 & 0.001 & 0.19 \\ \hline
\multicolumn{1}{|l||}{\multirow{6}{*}{Frobenius norm}} & \ac{S-COFI-PL} \ac{SCM} &  1.80  & 0.42  & 0.56 & 0.28 & 0.87 \\ \cline{2-7}
\multicolumn{1}{|l||}{}                  & \ac{COFI-PL} \ac{SCM} &  1.80 & 0.43 & 0.56 & 0.28 & 0.87 \\ \cline{2-7}
\multicolumn{1}{|l||}{}                  & \ac{S-COFI-PL} \ac{PO} &  1.31   &  0.22 & 0.74  &  0.61 & 0.94 \\ \cline{2-7}
\multicolumn{1}{|l||}{}                  & \ac{COFI-PL} \ac{PO} & 1.32 & 0.23 & 0.74 & 0.60 & 0.94 \\ \cline{2-7}
\multicolumn{1}{|l||}{}                  & \ac{S-COFI-PL} \ac{SK-SCM} &  1.80  & 0.42 &  0.56 & 0.28  & 0.87 \\ \cline{2-7}
\multicolumn{1}{|l||}{}                  & \ac{COFI-PL} \ac{SK-SCM} &  1.80 & 0.43 & 0.56 & 0.28 & 0.87 \\ \cline{2-7}
\multicolumn{1}{|l||}{}                  & \ac{S-COFI-PL} \ac{SK-PO} &  1.31  & 0.22  & 0.74 & 0.61 & 0.94 \\ \cline{2-7}
\multicolumn{1}{|l||}{}                  & \ac{COFI-PL} \ac{SK-PO} & 1.32 & 0.23 & 0.74 & 0.60 & 0.94 \\ \cline{2-7}
\multicolumn{1}{|l||}{}                  & \ac{S-COFI-PL} \ac{BW-SCM} &  1.76 & 0.41  & 0.57  & 0.30 & 0.85 \\ \cline{2-7}
\multicolumn{1}{|l||}{}                  & \ac{COFI-PL} \ac{BW-SCM} &  1.76 & 0.41 & 0.57 &  0.30 & 0.85 \\ \cline{2-7}
\multicolumn{1}{|l||}{}                  & \ac{S-COFI-PL} \ac{BW-PO} &  0.57 & 0.04 & 0.92 & 0.91 & 0.92 \\ \cline{2-7}
\multicolumn{1}{|l||}{}                  & \ac{COFI-PL} \ac{BW-PO} & \texttimes & \texttimes & \texttimes & \texttimes & 0.91 \\ \hline
\multicolumn{2}{|l||}{Sequential Estimator \cite{ansari2017sequential}} & 1.82  & 0.44 & 0.53 & 0.24 & 0.81 \\ \hline
\multicolumn{2}{|l||}{S-MLE-PL \cite{elhajjar2024}} &  2.34 & 0.74 & 0.18 & 0.01 & 0.33 \\ \hline
\end{tabular}
\caption{Quantitative evaluation and comparison of different approaches}
\label{tab:quantitative}
\end{table*}

\subsection{Computation time comparison}

\begin{table}[H]
\centering
\begin{tabular}{|cc|c|}
\hline
\multicolumn{2}{|c|}{Method}  & Computation time \\ \hline
\multicolumn{2}{|c|}{S-MLE-PL \cite{elhajjar2024}}  & $2.18$h \\ \hline 
\multicolumn{2}{|c|}{Sequential Estimator \cite{ansari2017sequential}} & $0.6$h \\ \hline 
\multicolumn{1}{|c|}{\multirow{2}{*}{\ac{S-COFI-PL}}} & \ac{KL} divergence   &  $0.44$h \\ \cline{2-3}  
\multicolumn{1}{|c|}{} & Frobenius norm  & $0.37$h \\ \hline 
\multicolumn{1}{|c|}{\multirow{2}{*}{\ac{COFI-PL} \cite{vu2023covariance} }} & \ac{KL} divergence & $0.52$h \\ \cline{2-3}  
\multicolumn{1}{|c|}{} &  Frobenius norm & $0.43$h \\ \hline 
\end{tabular}
\caption{Computation time comparison for S-MLE-PL \cite{elhajjar2024}, Sequential Estimator \cite{ansari2017sequential}, \ac{S-COFI-PL} and \ac{COFI-PL} for both \ac{KL} divergence and the Frobenius norm}
\label{tab:computation}
\end{table}

As shown in the complexity comparison (section \ref{subsec:complexity}), the proposed approach is less costly than the sequential approaches proposed in the state-of-the-art as well as the offline approaches. In this section, we compare the computation time of each approach applied in Mexico City data set of $40$ images of size $(3638, 16709)$ pixels.
We use a machine with a $95$-core CPU running at $2.2$ GHz and $125$ G of RAM with calculations executed in parallel across the CPUs. We show the computation time of each approach in the table \ref{tab:computation}.

\noindent
For a comparison between the offline processing and our approach, the latter is faster than the COFI-PL  \cite{vu2023covariance} offline algorithm for KL ($17\%$) and Frobenius norm ($13\%$). In comparison with the sequential approaches in the state of the art, our approach is faster than the S-MLE-PL method \cite{elhajjar2024} ($78\%$ and $83\%$ for \ac{KL} and Frobenius, respectively), as well as the sequential estimator \cite{ansari2017sequential} ($63\%$ and $68\%$ for \ac{KL} and Frobenius, respectively).

%% file: Sections/acknowledgment.tex
\section*{Acknowledgment}
This work is funded by the ANR REPED-SARIX project (ANR-21-CE23-0012-01) of the French national Agency of research.

%% file: Sections/conclusion.tex
\section{Conclusion}

In this paper, we propose a novel sequential approach based on covariance fitting interferometric \ac{PL}. The proposed framework accommodates the \ac{KL} divergence and the Frobenius norm as well as various plug-in for the covariance matrix estimation, which, through different regularization methods, ensure the robustness of the approach. We provide a \ac{MM} algorithm to solve the different optimization problems. Numerical experiments and real-world study shows the efficiency of the proposed approach in terms of performances and computation time.

%% file: biblio.tex
\ifCLASSOPTIONcaptionsoff
  \newpage
\fi

\bibliographystyle{IEEEbib}
\bibliography{biblio.bib}

%% file: Sections/appendix/appendix_2cols.tex
\appendices
\input{Sections/appendix/appendix_cost_fct_2cols}

\input{Sections/appendix/appendix_MM_2cols}

%% file: Sections/appendix/appendix_cost_fct_2cols.tex
\section{Calculation details of the cost functions}
\label{app_CF}

\noindent
We recall the block structure of the covariance matrix 
\begin{align*}
    \mathbf{\widetilde{\Sigma}}
    &= \left( \begin{array}{c}
        \begin{tabular}{cc}
           $\mathbf{\Sigma}_{p}$ & $(\mathbf{\Sigma}_{pn})^H$ \\
                $\mathbf{\Sigma}_{pn}$ &  $\mathbf{\Sigma}_{n}$ \\
        \end{tabular}
    \end{array} \right)
\end{align*}
where the different notations are provided in Table \ref{tab:tab_plug_in_app}.
\begin{table}[!h]
\renewcommand{\arraystretch}{1.3} 
\centering
\begin{tabular}{|c|c|c|c|}
\hline
\multirow{2}{*}{} &  \makecell{unconstrained  \\ plug-in  \\ (section \ref{subsubsec:unconstrained_seq})} 
                  & \makecell{shrinkage to \\  identity \\ regularization \\ (section \ref{subsubsec:shrinkage_seq})} 
                  &  \makecell{tapering  \\ regularization  \\ (section \ref{subsubsec:tapering_seq})} \\ \cline{2-4}
                  &  \acs{SCM} $\setminus$ \acs{PO} & \acs{SK-SCM} $\setminus$ \acs{SK-PO} &  \acs{BW-SCM} $\setminus$ \acs{BW-PO} \\ \hline
$\mathbf{\Sigma}_{p}$ & $\mathbf{\Sigma}_{p}^{U}$ & $\mathbf{\Sigma}_{p}^{SK}$ &  $\mathbf{\Sigma}_{p}^{BW}$ \\ \hline
$\mathbf{\Sigma}_{pn}$         & $\mathbf{\Sigma}_{pn}^{U}$ & $\mathbf{\Sigma}_{pn}^{SK}$ & $\mathbf{\Sigma}_{pn}^{BW}$ \\ \hline
$\mathbf{\Sigma}_{n}$    & $\mathbf{\Sigma}_{n}^{U}$ & $\mathbf{\Sigma}_{n}^{SK}$ & $\mathbf{\Sigma}_{n}^{BW}$ \\ \hline
\end{tabular}
\caption{$\mathbf{\Sigma}_{p}$, $\mathbf{\Sigma}_{pn}$ and $\mathbf{\Sigma}_{n}$ values depending on the choice of the covariance matrix plug-in.}
\label{tab:tab_plug_in_app}
\end{table}

\noindent
In this appendix, we provide the various calculation steps to derive the generic form of the cost function for the \acs{KL} divergence and the Frobenius norm respectively.

\subsection{Kullback-Leibler divergence}
\label{app_CF_KL}


\noindent
By  designating 
\begin{itemize}
    \item[-] $\mathbf{F} = |\mathbf{\Sigma}_{p}| - (|\mathbf{\Sigma}_{pn}|)^T |\mathbf{\Sigma}_{n}|^{-1}|\mathbf{\Sigma}_{pn}|$
    \item[-] $\mathbf{D} = |\mathbf{\Sigma}_{n}| - |\mathbf{\Sigma}_{pn}||\mathbf{\Sigma}_{p}|^{-1} (|\mathbf{\Sigma}_{pn}|)^T$
    \item[-] $\mathbf{A} = - \mathbf{D}^{-1} |\mathbf{\Sigma}_{pn}||\mathbf{\Sigma}_{p}|^{-1}$
    \item[-] $\mathbf{M} = \mathbf{D}^{-1} \circ \mathbf{\Sigma}_{n}$
\end{itemize}
\noindent
the inverse of a block structured matrix \cite{petersen2008matrix} has the following form 
\begin{align*}
    \mathbf{\widetilde{\Psi}}^{-1} &= 
    \begin{pmatrix}
        \mathbf{F}^{-1} & (\mathbf{A})^T \\
        \mathbf{A} & \mathbf{D}^{-1} \\
    \end{pmatrix}
\end{align*}
\noindent
The cost function in equation  (\ref{eq:CF_KL}) for the \acs{KL} divergence is obtained as follow

\noindent
\begin{minipage}{0.5\columnwidth}
\begin{align*}
    &f^{\text{KL}}_{\mathbf{\widetilde{\Sigma}}}(\mathbf{\widetilde{w}}_{\theta}) = \mathbf{\widetilde{w}}_{\theta}^H (\mathbf{\widetilde{\Psi}}^{-1} \circ \mathbf{\widetilde{\Sigma}}) \mathbf{\widetilde{w}}_{\theta} \\
    &= \begin{pmatrix}
        \mathbf{w}_{\theta}  \\
        \mathbf{\Bar{w}}_{\theta}
    \end{pmatrix} ^H 
    \begin{pmatrix}
    \begin{pmatrix}
        \mathbf{F}^{-1} & (\mathbf{A})^T \\
        \mathbf{A} & \mathbf{D}^{-1} \\
    \end{pmatrix}
    \circ 
    \begin{pmatrix}
      \mathbf{\Sigma}_{p} &   (\mathbf{\Sigma}_{pn})^H \\
      \mathbf{\Sigma}_{pn} & \mathbf{\Sigma}_{n} \\
    \end{pmatrix}
    \end{pmatrix} 
    \begin{pmatrix}
        \mathbf{w}_{\theta} \\
        \mathbf{\Bar{w}}_{\theta}
        \end{pmatrix} \\
    &= \begin{pmatrix}
        \mathbf{w}_{\theta} \\
        \mathbf{\Bar{w}}_{\theta} 
    \end{pmatrix}^H 
    \begin{pmatrix}
       \mathbf{F}^{-1} \circ \mathbf{\Sigma}_{p} & (\mathbf{A})^T  \circ  (\mathbf{\Sigma}_{pn})^H \\
        \mathbf{A} \circ  \mathbf{\Sigma}_{pn} &  \mathbf{M} \\
    \end{pmatrix}
    \begin{pmatrix}
        \mathbf{w}_{\theta} \\
        \mathbf{\Bar{w}}_{\theta}
        \end{pmatrix} \\
    &=  \begin{pmatrix}
        \mathbf{w}_{\theta}  \\
        \mathbf{\Bar{w}}_{\theta} 
    \end{pmatrix} ^H 
    \begin{pmatrix}
       \left( \mathbf{F}^{-1} \circ \mathbf{\Sigma}_{p}\right) \mathbf{w}_{\theta} + \left(  (\mathbf{A})^T  \circ  (\mathbf{\Sigma}_{pn})^H \right)\mathbf{\Bar{w}}_{\theta} \\
        \left(  \mathbf{A} \circ  \mathbf{\Sigma}_{pn} \right) \mathbf{w}_{\theta} +  \mathbf{M} \mathbf{\Bar{w}}_{\theta}\\
    \end{pmatrix} \\
    &= \mathbf{w}_{\theta}^H \left( \mathbf{F}^{-1} \circ \mathbf{\Sigma}_{p}\right) \mathbf{w}_{\theta} + \mathbf{w}_{\theta}^H \left(  (\mathbf{A})^T \circ  (\mathbf{\Sigma}_{pn})^H \right)\mathbf{\Bar{w}}_{\theta} \\
    &\qquad+ \mathbf{\Bar{w}}_{\theta}^H \left(  \mathbf{A} \circ  \mathbf{\Sigma}_{pn} \right) \mathbf{w}_{\theta} + \mathbf{\Bar{w}}_{\theta}^H  \mathbf{M}  \mathbf{\Bar{w}}_{\theta}
\end{align*}
\end{minipage}

\subsection{Frobenius norm}
\label{app_CF_LS}

\noindent
The generic cost function in equation (\ref{eq:CF_Frob}) for the Frobenius norm is obtained as follow

\begin{align*}
    f^{\text{FN}}_{\mathbf{\widetilde{\Sigma}}}(\mathbf{\widetilde{w}}_{\theta}) &= -2 \mathbf{\widetilde{w}}_{\theta}^H (\mathbf{\widetilde{\Psi}} \circ \mathbf{\widetilde{\Sigma}}) \mathbf{\widetilde{w}}_{\theta} \\
    &= -2 \begin{pmatrix}
        \mathbf{w}_{\theta} \\
        \mathbf{\Bar{w}}_{\theta}
    \end{pmatrix}^H  \\
    &\quad
    \left( 
    \ \begin{pmatrix}
            |\mathbf{\Sigma}_{p}| & (|\mathbf{\Sigma}_{pn}|)^T \\
            |\mathbf{\Sigma}_{pn}| &  |\mathbf{\Sigma}_{n}| \\
    \end{pmatrix} \circ 
    \begin{pmatrix}
      \mathbf{\Sigma}_{p} &   (\mathbf{\Sigma}_{pn})^H \\
      \mathbf{\Sigma}_{pn} & \mathbf{\Sigma}_{n} \\
    \end{pmatrix}
    \right) 
    \begin{pmatrix}
        \mathbf{w}_{\theta} \\
        \mathbf{\Bar{w}}_{\theta}
    \end{pmatrix} \\
    &= -2 \begin{pmatrix}
        \mathbf{w}_{\theta} \\
        \mathbf{\Bar{w}}_{\theta}
    \end{pmatrix}^H \\
    &\quad
    \begin{pmatrix}
        (|\mathbf{\Sigma}_{p}| \circ \mathbf{\Sigma}_{p}) \mathbf{w}_{\theta} + (|\mathbf{\Sigma}_{pn}|)^T \circ (\mathbf{\Sigma}_{pn})^H ) \mathbf{\Bar{w}}_{\theta} \\
        (|\mathbf{\Sigma}_{pn}| \circ \mathbf{\Sigma}_{pn} )\mathbf{w}_{\theta} + (|\mathbf{\Sigma}_{n}| \circ \mathbf{\Sigma}_{n}) \mathbf{\Bar{w}}_{\theta}
    \end{pmatrix} \\
    &= - 2 \mathbf{w}_{\theta}^H (|\mathbf{\Sigma}_{p}| \circ \mathbf{\Sigma}_{p}) \mathbf{w}_{\theta} \\
    &\qquad- 2 \mathbf{w}_{\theta}^H (|\mathbf{\Sigma}_{pn}|)^T \circ (\mathbf{\Sigma}_{pn})^H ) \mathbf{\Bar{w}}_{\theta} \\
    &\qquad - 2 \mathbf{\Bar{w}}_{\theta}^H (|\mathbf{\Sigma}_{pn}| \circ \mathbf{\Sigma}_{pn} ) \mathbf{w}_{\theta} \\
    &\qquad - 2 \mathbf{\Bar{w}}_{\theta}^H (|\mathbf{\Sigma}_{n}| \circ \mathbf{\Sigma}_{n}) \mathbf{\Bar{w}}_{\theta}
\end{align*}

%% file: Sections/appendix/appendix_MM_2cols.tex
\section{Calculations details for the \texorpdfstring{\ac{MM}}{MM} algorithm}
\label{app_MM}

In this appendix, we present the steps to derive the surrogate function for each of the distances, the \acs{KL} divergence and the Frobenius norm respectively.

\subsection{Kullback-Leibler divergence}
\label{app_MM_KL}

The cost function in equation (\ref{eq:CF_KL})  can be reformulated as follow 
\begin{align*}
   f^{\text{KL}}_{\mathbf{\widetilde{\Sigma}}}(\mathbf{\widetilde{w}}_{\theta}) &= \mathbf{w}_{\theta}^H \left( \mathbf{F}^{-1} \circ \mathbf{\Sigma}_p \right) \mathbf{w}_{\theta} \\
   &\qquad + \mathbf{w}_{\theta}^H \left( \left( - |\mathbf{\Sigma}_{p}|^{-1} (|\mathbf{\Sigma}_{pn}|)^T  \mathbf{D}^{-1} \right) \circ  (\mathbf{\Sigma}_{pn})^H \right)\mathbf{\Bar{w}}_{\theta} \\
   &\qquad + \mathbf{\Bar{w}}_{\theta}^H \left( \left( - \mathbf{D}^{-1} |\mathbf{\Sigma}_{pn}||\mathbf{\Sigma}_{p}|^{-1}\right)\circ  \mathbf{\Sigma}_{pn} \right) \mathbf{w}_{\theta} \\
   &\qquad + \mathbf{\Bar{w}}_{\theta}^H \mathbf{M} \mathbf{\Bar{w}}_{\theta}\\
    &\propto \left(\mathbf{\Bar{w}}_{\theta}^H \left( \left( - \mathbf{D}^{-1} |\mathbf{\Sigma}_{pn}||\mathbf{\Sigma}_{p}|^{-1}\right)\circ  \mathbf{\Sigma}_{pn} \right) \mathbf{w}_{\theta} \right)^H \\
    &\qquad + \mathbf{\Bar{w}}_{\theta}^H \left( \left( - \mathbf{D}^{-1} |\mathbf{\Sigma}_{pn}||\mathbf{\Sigma}_{p}|^{-1}\right)\circ  \mathbf{\Sigma}_{pn} \right) \mathbf{w}_{\theta} \\
    &\qquad + \mathbf{\Bar{w}}_{\theta}^H \mathbf{M} \mathbf{\Bar{w}}_{\theta}\\
    &\propto 2 \textit{Re}(\mathbf{\Bar{w}}_{\theta}^H \left( \left( - \mathbf{D}^{-1} |\mathbf{\Sigma}_{pn}||\mathbf{\Sigma}_{p}|^{-1}\right)\circ  \mathbf{\Sigma}_{pn} \right) \mathbf{w}_{\theta}) \\
    &\qquad + \mathbf{\Bar{w}}_{\theta}^H \mathbf{M} \mathbf{\Bar{w}}_{\theta} 
\end{align*}
Using Lemma \ref{lemma1}, the above cost function can be majorized by the following expression
\begin{align*}
    f^{\text{KL}}_{\mathbf{\widetilde{\Sigma}}}(\mathbf{\widetilde{w}}_{\theta}) &\leq 2 \text{Re}\left(\mathbf{\Bar{w}}_{\theta}^H \left( \left( - \mathbf{D}^{-1} |\mathbf{\Sigma}_{pn}||\mathbf{\Sigma}_{p}|^{-1}\right)\circ  \mathbf{\Sigma}_{pn} \right) \mathbf{w}_{\theta}\right)  \\
    &\qquad + 2 \text{Re}\left(\mathbf{\Bar{w}}_{\theta}^H \left[ \mathbf{M} - \lambda_{\max}^{\mathbf{M}} \mathbf{I}_k \right] \mathbf{\Bar{w}}_{\theta}^{(t)} \right) \\
    &= \text{Re}\left( 2 \mathbf{\Bar{w}}_{\theta}^H \left( \left( - \mathbf{D}^{-1} |\mathbf{\Sigma}_{pn}||\mathbf{\Sigma}_{p}|^{-1}\right)\circ  \mathbf{\Sigma}_{pn} \right) \mathbf{w}_{\theta}  \right. \\
    &\left. \qquad + 2  \mathbf{\Bar{w}}_{\theta}^H \left[  \mathbf{M} - \lambda_{\max}^{\mathbf{M}} \mathbf{I}_k \right] \mathbf{\Bar{w}}_{\theta}^{(t)} \right) \\
    &= - \text{Re}\left( \mathbf{\Bar{w}}_{\theta}^H  2 \left[ \left( \left( \mathbf{D}^{-1} |\mathbf{\Sigma}_{pn}||\mathbf{\Sigma}_{p}|^{-1}\right)\circ  \mathbf{\Sigma}_{pn} \right) \mathbf{w}_{\theta} \right. \right. \\
    &\left. \left. \qquad - \left[  \mathbf{M} - \lambda_{\max}^{\mathbf{M}} \mathbf{I}_k \right] \mathbf{\Bar{w}}_{\theta}^{(t)} \right] \right)
\end{align*}

\subsection{Frobenius norm}
\label{app_MM_LS}

\noindent
The cost function in equation (\ref{eq:CF_Frob}) can be written as follow
\begin{align*}
    f^{\text{FN}}_{\mathbf{\widetilde{\Sigma}}}(\mathbf{\widetilde{w}}_{\theta}) &=  - 2 \mathbf{w}_{\theta}^H (|\mathbf{\Sigma}_{p}| \circ \mathbf{\Sigma}_p) \mathbf{w}_{\theta}  \\
    &\qquad - 2 \mathbf{w}_{\theta}^H (|\mathbf{\Sigma}_{pn}|^T \circ (\mathbf{\Sigma}_{pn})^H ) \mathbf{\Bar{w}}_{\theta} \\
    &\qquad - 2 \mathbf{\Bar{w}}_{\theta}^H (|\mathbf{\Sigma}_{pn}| \circ \mathbf{\Sigma}_{pn} ) \mathbf{w}_{\theta} \\
    &\qquad - 2 \mathbf{\Bar{w}}_{\theta}^H (|\mathbf{\Sigma}_{n}| \circ \mathbf{\Sigma}_{n}) \mathbf{\Bar{w}}_{\theta} \\
    &\propto - 2 (\mathbf{\Bar{w}}_{\theta}^H (|\mathbf{\Sigma}_{pn}| \circ \mathbf{\Sigma}_{pn} ) \mathbf{w}_{\theta})^H \\
    &\qquad - 2 \mathbf{\Bar{w}}_{\theta}^H (|\mathbf{\Sigma}_{pn}| \circ \mathbf{\Sigma}_{pn} ) \mathbf{w}_{\theta} \\
    &\qquad - 2 \mathbf{\Bar{w}}_{\theta}^H (|\mathbf{\Sigma}_{n}| \circ \mathbf{\Sigma}_{n}) \mathbf{\Bar{w}}_{\theta} \\
    &\propto - 4 \textit{Re}(\mathbf{\Bar{w}}_{\theta}^H (|\mathbf{\Sigma}_{pn}| \circ \mathbf{\Sigma}_{pn} ) \mathbf{w}_{\theta}) \\
    &\qquad - 2 \mathbf{\Bar{w}}_{\theta}^H (|\mathbf{\Sigma}_{n}| \circ \mathbf{\Sigma}_{n}) \mathbf{\Bar{w}}_{\theta}
\end{align*}

\noindent
Using Lemma \ref{lemma3}, the cost function can be majorized by the following expression
\begin{align*}
    f^{\text{FN}}_{\mathbf{\widetilde{\Sigma}}}(\mathbf{\widetilde{w}}_{\theta}) &\leq - 4 \textit{Re}\left(\mathbf{\Bar{w}}_{\theta}^H (|\mathbf{\Sigma}_{pn}| \circ \mathbf{\Sigma}_{pn} ) \mathbf{w}_{\theta}\right) \\
    &\qquad - 4 \text{Re}\left(\mathbf{\Bar{w}}_{\theta}^H (|\mathbf{\Sigma}_{n}| \circ \mathbf{\Sigma}_{n}) \mathbf{\Bar{w}}_{\theta}^{(t)} \right)\\
    &= - 4 \textit{Re}\left(\mathbf{\Bar{w}}_{\theta}^H (|\mathbf{\Sigma}_{pn}| \circ \mathbf{\Sigma}_{pn} ) \mathbf{w}_{\theta} \right. \\
    &\left. \qquad + \mathbf{\Bar{w}}_{\theta}^H (|\mathbf{\Sigma}_{n}| \circ \mathbf{\Sigma}_{n}) \mathbf{\Bar{w}}_{\theta}^{(t)}\right) \\
    &= - 4 \textit{Re}\left(\mathbf{\Bar{w}}_{\theta}^H \left[(|\mathbf{\Sigma}_{pn}| \circ \mathbf{\Sigma}_{pn} ) \mathbf{w}_{\theta} \right. \right.\\
    &\left. \left. \qquad +  (|\mathbf{\Sigma}_{n}| \circ \mathbf{\Sigma}_{n}) \mathbf{\Bar{w}}_{\theta}^{(t)} \right] \right) \\
    &= - \textit{Re}\left(\mathbf{\Bar{w}}_{\theta}^H  4 \left[(|\mathbf{\Sigma}_{pn}| \circ \mathbf{\Sigma}_{pn} ) \mathbf{w}_{\theta} \right. \right. \\
    &\left. \left. \qquad +  (|\mathbf{\Sigma}_{n}|\circ \mathbf{\Sigma}_{n}) \mathbf{\Bar{w}}_{\theta}^{(t)} \right] \right) \\
\end{align*} 